\documentclass{article}
\usepackage{graphicx,verbatim, bussproofs,tikz,float, latexsym}
\usetikzlibrary{arrows}
\usepackage{pgfplots,multicol}
\usepackage{amsmath,amsthm}
\usepackage{amssymb}
\usepackage{stmaryrd}
\usepackage{proof}
\usepackage{cmll}
\usepackage{bussproofs}
\usepackage{authblk}
\DeclareSymbolFont{extraup}{U}{zavm}{m}{n}
\DeclareMathSymbol{\vardiamond}{\mathalpha}{extraup}{87}

\newcommand{\nomi}{\mathbf{i}}
\newcommand{\nomj}{\mathbf{j}}
\newcommand{\nomk}{\mathbf{k}}

\newcommand{\bigamp}{\mathop{\mbox{\Large \&}}}

\renewcommand{\phi}{\varphi}
\renewcommand{\emptyset}{\varnothing}

\newcommand{\Diamondblack}{\vardiamond}

\renewcommand{\epsilon}{\varepsilon}
\theoremstyle{definition}
\newtheorem{theorem}{Theorem}[section]
\newtheorem{lemma}[theorem]{Lemma}

\newtheorem{proposition}[theorem]{Proposition}

\newtheorem{example}[theorem]{Example}

\newtheorem{definition}[theorem]{Definition}

\title{Algorithmic Correspondence for Hybrid Logic with Binder}
\author{Zhiguang Zhao}

\affil{\small School of Mathematics and Statistics, Taishan University, Tai'an, 271000, China}
\affil{\small zhaozhiguang23@gmail.com}

\date{}

\begin{document}
\maketitle
\begin{abstract}
In the present paper, we develop the algorithmic correspondence theory for hybrid logic with binder $\mathcal{H}(@, \downarrow)$. We define the class of Sahlqvist inequalities for $\mathcal{H}(@, \downarrow)$, each inequality of which is shown to have a first-order frame correspondent effectively computable by an algorithm $\mathsf{ALBA}^{\downarrow}$. 

\end{abstract}

\section{Introduction}

\paragraph{Hybrid Logic}Hybrid logics \cite{BeBlWo06} refer to a number of extensions of modal logic where it is possible to refer to states by the so-called \emph{nominals} which are true at exactly one world. In addition, there are different connectives that can be added to the hybrid language to further extend the expressive power. Two such examples are the \emph{satisfaction operator} $@_{\mathbf{i}}$ which allows one to jump to the world denoted by the nominal $\mathbf{i}$, and the \emph{binder} $\downarrow x$ which binds the current world and can refer to the world later in the formula. 

\paragraph{Correspondence Theory}Correspondence theory concerns the relation between modal formulas and first-order formulas. We say that a modal formula and a first-order formula correspond to each other if they define the same class of Kripke frames. Early results concerning correspondence theory are Sahlqvist's \cite{Sa75} and van Benthem's \cite{vB83}, who gave a syntactic characterization of certain modal formulas (later called \emph{Sahlqvist formulas}) which have first-order correspondents and they are canonical, which implies that any normal modal logic axiomatized with Sahlqvist formulas is strongly complete with respect to the class of all Kripke frames satisfying the corresponding first-order conditions. The Sahlqvist-van Benthem algorithm \cite{Sa75,vB83} was given to transform a Sahlqvist formula into its first-order correspondent.

\paragraph{Correspondence Theory for Hybrid Logic}In the literature, there are many existing works on the correspondence theory for hybrid logic \cite{BetCMaVi04,Co09,
CoGoVa06b,ConRob,GaGo93,GoVa01,Ho06,HoPa10,Ta05,tCMaVi06}. In particular, ten Cate et al.\ \cite{tCMaVi06} showed that any hybrid logic obtained by adding modal Sahlqvist formulas to the basic hybrid logic H is strongly complete. Gargov and Goranko showed that any extension of H with pure axioms (formulas containing no propositional variables but only possibly nominals) is strongly complete. In \cite{tCMaVi06} it was shown that these two results cannot be combined in general, since there is a
modal Sahlqvist formula and a pure formula which together give a Kripke-incomplete logic when added to H. Conradie and Robinson \cite{ConRob} investigated to what extent these two results can be combined. In the end of \cite{ConRob}, it was mentioned that a further direction would be to extend results concerning extending correspondence theory to more expressive hybrid languages e.g.\ hybrid logic with binder, which is the focus of the present paper.

\paragraph{Unified correspondence} The present paper belongs to the theory of \emph{unified correspondence} \cite{ConPalSou,CoGhPa14}. One major part of this theory is the algorithm $\mathsf{ALBA}$ (Ackermann Lemma Based Algorithm), which computes the first-order correspondents of input formulas/inequalities and is guaranteed to succeed on Sahlqvist inequalities.

\paragraph{Structure of the paper}
In the present paper, we will use the algorithmic methodology to provide a correspondence theory treatment of hybrid logic with binder. Section \ref{Sec:Prelim} presents preliminaries on hybrid logic with binder, including syntax and semantics. Section \ref{Sec:Prelim:ALBA} provides preliminaries on algorithmic correspondence theory. Section \ref{Sec:Sahl} defines Sahlqvist inequalities. Section \ref{Sec:ALBA} gives the Ackermann Lemma Based Algorithm ($\mathsf{ALBA}^{\downarrow}$) for hybrid logic with binder. Section \ref{Sec:Soundness} gives the soundness proof of the algorithm. Section \ref{Sec:Success} shows that $\mathsf{ALBA}^{\downarrow}$ succeeds on Sahlqvist inequalities. Section \ref{Sec:Conclusion} gives conclusions.

\section{Preliminaries on hybrid logic with binder}\label{Sec:Prelim}

In the present section we collect the preliminaries on hybrid logic with binder. For more details, see \cite[Chapter 14]{BeBlWo06}.

\subsection{Language and Syntax}\label{Subsec:Lan:Syn}

\begin{definition}

Given countably infinite sets $\mathsf{Prop}$ of propositional variables, $\mathsf{Nom}$ of nominals, $\mathsf{Svar}$ of state variables, which are pairwise disjoint, the hybrid language $\mathcal{H}(@,\downarrow)$ is defined as follows:

$$\varphi::=p \mid \bot \mid \top \mid \mathbf{i} \mid x \mid \neg\varphi \mid (\varphi\land\varphi) \mid (\varphi\lor\varphi) \mid (\varphi\to\varphi) \mid \Box\varphi \mid \Diamond\varphi \mid @_{\mathbf{i}}\varphi \mid @_{x}\varphi\mid \downarrow x.\varphi,$$
where $p\in \mathsf{Prop}$, $\mathbf{i}\in\mathsf{Nom}$, $x\in\mathsf{Svar}$. We will follow the standard rules for omission of the parentheses. We also use $\mathsf{Prop}(\phi)$ to denote the propositional variables occuring in $\phi$. We use the notation $\vec p$ to denote a set of propositional variables and use $\phi(\vec p)$ to indicate that the propositional variables occur in $\phi$ are all in $\vec p$. We call a formula \emph{pure} if it does not contain propositional variables. In the present article we will consider only the hybrid language with one unary modality.
\end{definition}

Throughout the article, we will also make substantial use of the following expressions:

\begin{definition}
$\ $
\begin{itemize}
\item An \emph{inequality} is of the form $\phi\leq\psi$, where $\phi$ and $\psi$ are formulas.

\item A \emph{quasi-inequality} is of the form $\phi_1\leq\psi_1\ \&\ \ldots\ \&\ \phi_n\leq\psi_n\ \Rightarrow\ \phi\leq\psi$. 

\item A \emph{Mega-inequality} is defined inductively as follows:
$$\mathsf{Mega}::=\mathsf{Ineq}\mid \mathsf{Mega}\bigamp \mathsf{Mega} \mid \forall x(\mathsf{Mega})$$

where $\mathsf{Ineq}$ is an inequality, $\bigamp$ is the meta-conjunction and $\forall x$ is a universal state quantifier.

\item A \emph{universally quantified inequality} is defined as $\forall x_1\ldots\forall x_n(\phi\leq\psi)$.

\item A \emph{quasi-universally quantified inequality} is defined as $\mathsf{UQIneq_1}\&\ldots\& \mathsf{UQIneq_n}\Rightarrow\mathsf{UQIneq}$ where $\mathsf{UQIneq}, \mathsf{UQIneq_i}$ are universally quantified inequalities.
\end{itemize}
\end{definition}

We will find it easy to work with inequalities $\phi\leq\psi$ in place of implicative formulas $\phi\to\psi$ in Section \ref{Sec:Sahl}.

\subsection{Semantics}\label{Subsec:Seman}

\begin{definition}

A \emph{Kripke frame} is a pair $\mathbb{F}=(W,R)$ where $W\neq\emptyset$ is the \emph{domain} of $\mathbb{F}$, the \emph{accessibility relation} $R$ is a binary relation on $W$. A \emph{pointed Kripke frame} is a pair $(\mathbb{F}, w)$ where $w\in W$.
A \emph{Kripke model} is a pair $\mathbb{M}=(\mathbb{F}, V)$ where $V:\mathsf{Prop}\cup\mathsf{Nom}\to P(W)$ is a \emph{valuation} on $\mathbb{F}$ such
that for all nominals $\nomi\in\mathsf{Nom}$, $V(\nomi)$ is a singleton subset of $W$.

An assignment $g$ on $\mathbb{M}=(W,R,V)$ is a map $g:\mathsf{Svar}\to W$. Given an assignment $g$,  $x\in\mathsf{Svar}$, $w\in W$, we can define $g^{x}_{w}$, the \emph{$x$-variant of $g$} as follows: $g^{x}_{w}(y)=g(y)$ for all $y\in\mathsf{Svar}\setminus\{x\}$ and $g^{x}_{w}(x)=w$.

Now the satisfaction relation can be defined as follows: given any Kripke model $\mathbb{M}=(W,R,V)$, any assignment $g$ on $\mathbb{M}$, any $w\in W$, 

\begin{center}
\begin{tabular}{l c l}
$\mathbb{M},g,w\Vdash p$ & iff & $w\in V(p)$;\\
$\mathbb{M},g,w\Vdash \bot$ & : & never;\\
$\mathbb{M},g,w\Vdash \top$ & : & always;\\
$\mathbb{M},g,w\Vdash\nomi$ & iff & $\{w\}=V(\nomi)$;\\
$\mathbb{M},g,w\Vdash x$ & iff & $g(x)=w$;\\
$\mathbb{M},g,w\Vdash \neg\varphi$ & iff & $\mathbb{M},g,w\nVdash\varphi$;\\
$\mathbb{M},g,w\Vdash\varphi\land\psi$ & iff & $\mathbb{M},g,w\Vdash \varphi$ and $\mathbb{M},g,w\Vdash\psi$;\\
$\mathbb{M},g,w\Vdash\varphi\lor\psi$ & iff & $\mathbb{M},g,w\Vdash \varphi$ or $\mathbb{M},g,w\Vdash\psi$;\\
$\mathbb{M},g,w\Vdash\varphi\to\psi$ & iff & $\mathbb{M},g,w\nVdash \varphi$ or $\mathbb{M},g,w\Vdash\psi$;\\
$\mathbb{M},g,w\Vdash \Box\varphi$ & iff & $\forall v(Rwv\ \Rightarrow\ \mathbb{M},g,v\Vdash\varphi)$;\\
$\mathbb{M},g,w\Vdash\Diamond\varphi$ & iff & $\exists v(Rwv\ \mbox{ and }\ \mathbb{M},g,v\Vdash\varphi)$;\\
$\mathbb{M},g,w\Vdash @_{\mathbf{i}}\varphi$ & iff & $\mathbb{M},g,V(\nomi)\Vdash\varphi$;\\
$\mathbb{M},g,w\Vdash @_{x}\varphi$ & iff & $\mathbb{M},g,g(x)\Vdash\varphi$;\\
$\mathbb{M},g,w\Vdash \downarrow x.\varphi$ & iff & $\mathbb{M},g^{x}_{w},w\Vdash\varphi$.\label{page:downarrow}\\
\end{tabular}
\end{center}
For any formula $\phi$, we let $\llbracket\varphi\rrbracket^{\mathbb{M},g}=\{w\in W\mid \mathbb{M},g,w\Vdash\varphi\}$ denote the \emph{truth set} of $\varphi$ in $(\mathbb{M},g)$. The formula $\varphi$ is \emph{globally true} on $(\mathbb{M},g)$ (notation: $\mathbb{M},g\Vdash\varphi$) if $\mathbb{M},g,w\Vdash\varphi$ for every $w\in W$. We say that $\varphi$ is \emph{valid} on a Kripke frame $\mathbb{F}$ (notation: $\mathbb{F}\Vdash\varphi$) if $\varphi$ is globally true on $(\mathbb{F},V,g)$ for every valuation $V$ and every assignment $g$.
\end{definition}

For the semantics of inequalities, quasi-inequalities, mega-inequalities, universally quantified inequalities, quasi-universally quantified inequalities, they are given as follows:

\begin{definition}
$\ $
\begin{itemize}
\item An inequality is interpreted as follows:
$$(W,R,V),g\Vdash\phi\leq\psi\mbox{ iff }$$$$(\mbox{for all }w\in W, \mbox{ if }(W,R,V),g,w\Vdash\phi, \mbox{ then }(W,R,V),g,w\Vdash\psi).$$

\item A quasi-inequality is interpreted as follows:
$$(W,R,V),g\Vdash\phi_1\leq\psi_1\ \&\ \ldots\ \&\ \phi_n\leq\psi_n\ \Rightarrow\ \phi\leq\psi\mbox{ iff }$$
$$(W,R,V),g\Vdash\phi\leq\psi\mbox{ holds whenever }(W,R,V),g\Vdash\phi_i\leq\psi_i\mbox{ for all }1\leq i\leq n.$$

\item A Mega-inequality is interpreted as follows:

\begin{itemize}
\item $(W,R,V),g\Vdash \mathsf{Ineq}$ iff the inequality holds as defined in the definition above;

\item $(W,R,V),g\Vdash \mathsf{Mega_{1}}\bigamp\mathsf{Mega_{2}}$ iff $(W,R,V),g\Vdash \mathsf{Mega_{1}}$ and $(W,R,V),g\Vdash \mathsf{Mega_{2}}$;

\item $(W,R,V),g\Vdash\forall x(\mathsf{Mega})$ iff $(W,R,V),g^{x}_{w}\Vdash\mathsf{Mega}$ for all $w\in W$.

\end{itemize}

\item A universally quantified inequality is interpreted as follows:

$(W,R,V),g\Vdash\forall x_1\ldots\forall x_n(\phi\leq\psi)$ iff for all $w_1, \ldots, w_n\in W$, $(W,R,V),g^{x_1,\ldots,x_n}_{w_1, \ldots, w_n}\Vdash\phi\leq\psi$;

\item A quasi-universally quantified inequality is interpreted as follows:

$$(W,R,V),g\Vdash\mathsf{UQIneq_1}\&\ldots\& \mathsf{UQIneq_n}\Rightarrow\mathsf{UQIneq}\mbox{ iff }$$
$$(W,R,V),g\Vdash\mathsf{UQIneq}\mbox{ holds whenever }(W,R,V),g\Vdash\mathsf{UQIneq_i}\mbox{ for all }1\leq i\leq n.$$

\end{itemize}
\end{definition}

The definitions of validity are similar to formulas. It is easy to see that $(W,R,V),g\Vdash\phi\leq\psi$ iff $(W,R,V),g\Vdash\phi\to\psi$.

\section{Preliminaries on algorithmic correspondence}\label{Sec:Prelim:ALBA}

In this section, we give necessary preliminaries on the correspondence algorithm $\mathsf{ALBA}^{\downarrow}$ for hybrid logic with binder in the style of \cite{CoGoVa06,CoPa12,Zh21}. The algorithm $\mathsf{ALBA}^{\downarrow}$ transforms the input hybrid inequality $\varphi(\vec p)\leq\psi(\vec p)$ into an equivalent set of pure quasi-(universally quantified) inequalities which does not contain occurrences of propositional variables, and therefore can be translated into the first-order correspondence language via the standard translation of the expanded language of hybrid logic with binder (see page \pageref{page:expanded:language}).

The ingredients for the algorithmic correspondence proof to go through can be listed as follows:

\begin{itemize}
\item An expanded hybrid modal language as the syntax of the algorithm, as well as its interpretations in the relational semantics;
\item An algorithm $\mathsf{ALBA}^{\downarrow}$ which transforms a given hybrid inequality $\varphi(\vec p)\leq\psi(\vec p)$ into equivalent pure quasi-(universally quantified) inequalities $\mathsf{Pure}(\varphi(\vec p)\leq\psi(\vec p))$;
\item A soundness proof of the algorithm;
\item A syntactically identified class of inequalities on which the algorithm is successful;
\item A first-order correspondence language and first-order translation which transforms pure quasi-(universally quantified) inequalities into their equivalent first-order correspondents.
\end{itemize}

In the remainder of the paper, we will define an expanded hybrid modal language which the algorithm will manipulate (Section \ref{Sub:expanded:language}), define the first-order correspondence language of the expanded hybrid modal language and the standard translation (Section \ref{Sub:FOL:ST}). We give the definition of Sahlqvist inequalities (Section \ref{Sec:Sahl}), define a modified version of the algorithm $\mathsf{ALBA}^{\downarrow}$ (Section \ref{Sec:ALBA}), and show its soundness (Section \ref{Sec:Soundness}) and success on Sahlqvist inequalities (Section \ref{Sec:Success}).

\subsection{The expanded hybrid modal language}\label{Sub:expanded:language}

In the present subsection, we give the definition of the expanded hybrid modal language \label{page:expanded:language}, which will be used in the execution of the algorithm:
$$\varphi::=p \mid \bot \mid \top \mid \mathbf{i} \mid x \mid \neg\varphi \mid (\varphi\land\varphi) \mid (\varphi\lor\varphi) \mid (\varphi\to\varphi) \mid \Box\varphi \mid \Diamond\varphi \mid @_{\mathbf{i}}\varphi \mid @_{x}\varphi\mid \downarrow x.\varphi \mid$$
$$\blacksquare\phi \mid \Diamondblack\phi \mid \mathsf{A}\phi \mid \mathsf{E}\phi \mid \forall x\phi \mid \exists x\phi$$

For $\blacksquare$ and $\Diamondblack$, they are interpreted as the box and diamond modality on the inverse relation $R^{-1}$. $\mathsf{A}$ and $\mathsf{E}$ are global box and diamond modalities respectively, $\forall x\phi$ indicates that for all $x$-variant $g^{x}_{v}$ of $g$, $(W,R,V),g^{x}_{v},w\Vdash\phi$, and $\exists x\phi$ is the corresponding existential statement.

For the semantics of the expanded hybrid modal language, the additional semantic clauses can be given as follows:
\begin{center}
\begin{tabular}{l c l}
$\mathbb{M},g,w\Vdash \blacksquare\phi$ & iff & for all $v$ s.t. $(v,w)\in R$, $\mathbb{M},g,v\Vdash \phi$\\
$\mathbb{M},g,w\Vdash \Diamondblack\phi$ & iff & there exists a $v$ s.t. $(v,w)\in R$ and $\mathbb{M},g,v\Vdash \phi$\\
$\mathbb{M},g,w\Vdash \mathsf{A}\phi$ & iff & for all $v\in W$, $\mathbb{M},g,v\Vdash\phi$\\
$\mathbb{M},g,w\Vdash \mathsf{E}\phi$ & iff & there exists a $v\in W$ s.t.\ $\mathbb{M},g,v\Vdash\phi$\\
$\mathbb{M},g,w\Vdash\forall x\phi$ & iff & for all $v\in W$, $\mathbb{M},g^{x}_{v},w\Vdash\phi$\\
$\mathbb{M},g,w\Vdash\exists x\phi$ & iff & there exists a $v\in W$ s.t.\ $\mathbb{M},g^{x}_{v},w\Vdash\phi$.\\
\end{tabular}
\end{center}

\subsection{The first-order correspondence language and the standard translation}\label{Sub:FOL:ST}

In the first-order correspondence language, we have a binary predicate symbol $R$ corresponding to the binary relation, a set of constant symbols $i$ corresponding to each nominal $\nomi$, a set of unary predicate symbols $P$ corresponding to each propositional variable $p$. The state variables $x$ correspond to individual variables $x$ in the first-order language.

\begin{definition}
The standard translation of the expanded hybrid modal language is defined as follows:

\begin{itemize}
\item $ST_{x}(p):=Px$;
\item $ST_{x}(\bot):=\bot$;
\item $ST_{x}(\top):=\top$;
\item $ST_{x}(\nomi):=x=i$;
\item $ST_{x}(y):=x=y$;
\item $ST_{x}(\neg\phi):=\neg ST_{x}(\phi)$;
\item $ST_{x}(\phi\land\psi):=ST_{x}(\phi)\land ST_{x}(\psi)$;
\item $ST_{x}(\phi\lor\psi):=ST_{x}(\phi)\lor ST_{x}(\psi)$;
\item $ST_{x}(\phi\to\psi):=ST_{x}(\phi)\to ST_{x}(\psi)$;
\item $ST_{x}(\Box\phi):=\forall y(Rxy\to ST_{y}(\phi))$ ($y$ does not occur in $\phi$);
\item $ST_{x}(\Diamond\phi):=\exists y(Rxy\land ST_{y}(\phi))$ ($y$ does not occur in $\phi$);
\item $ST_{x}(@_{\nomi}\phi):=\exists y(y=i\land ST_{y}(\phi))$ ($y$ does not occur in $\phi$);
\item $ST_{x}(@_{z}\phi):=\exists y(y=z\land ST_{y}(\phi))$ ($y$ does not occur in $\phi$);
\item $ST_{x}(\downarrow y.\phi):=\exists y(y=x\land ST_{x}(\phi))$;
\item $ST_{x}(\blacksquare\phi):=\forall y(Ryx\to ST_{y}(\phi))$ ($y$ does not occur in $\phi$);
\item $ST_{x}(\Diamondblack\phi):=\exists y(Ryx\land ST_{y}(\phi))$ ($y$ does not occur in $\phi$);
\item $ST_{x}(\mathsf{A}\phi):=\forall y ST_{y}(\phi)$ ($y$ does not occur in $\phi$);
\item $ST_{x}(\mathsf{E}\phi):=\exists y ST_{y}(\phi)$ ($y$ does not occur in $\phi$);
\item $ST_{x}(\forall y\phi):=\forall y ST_{x}(\phi)$;
\item $ST_{x}(\exists y\phi):=\exists y ST_{x}(\phi)$.
\end{itemize}
\end{definition}

It is easy to see that this translation is correct:

\begin{proposition}
For any Kripke model $\mathbb{M}$, any assignment $g$ on $\mathbb{M}$, any $w\in W$ and any expanded hybrid modal formula $\phi$, 
$$\mathbb{M},g,w\Vdash\phi\mbox{ iff }\mathbb{M},g^{x}_{w}\vDash ST_{x}(\phi),$$
where $x$ is a fresh variable not occurring in $\phi$.
\end{proposition}

For inequalities, quasi-inequalities, mega-inequalities, universally quantified inequalities and quasi-universally quantified inequalities, the standard translation is given in a global way:

\begin{definition}
\begin{itemize}
\item $ST(\phi\leq\psi):=\forall x(ST_{x}(\phi)\to ST_{x}(\psi))$ ($x$ does not occur in $\phi$ and $\psi$);
\item $ST(\phi_1\leq\psi_1\&\ldots\&\phi_n\leq\psi_n\Rightarrow\phi\leq\psi):= ST(\phi_1\leq\psi_1)\land\ldots\land ST(\phi_n\leq\psi_n)\to ST(\phi\leq\psi)$;
\item $ST(\mathsf{Mega}_1\ \&\ \mathsf{Mega}_2)=ST(\mathsf{Mega}_1)\land ST(\mathsf{Mega}_2)$;
\item $ST(\forall x(\mathsf{Mega})):=\forall x ST(\mathsf{Mega})$;
\item $ST(\forall x_1\ldots\forall x_n\mathsf{Ineq}):=\forall x_1\ldots\forall x_n ST(\mathsf{Ineq})$;
\item $ST(\mathsf{UQIneq_1}\&\ldots\&\mathsf{UQIneq_n}\Rightarrow\mathsf{UQIneq}):= ST(\mathsf{UQIneq_1})\land\ldots\land ST(\mathsf{UQIneq_n})\to ST(\mathsf{UQIneq})$.
\end{itemize}
\end{definition}

\begin{proposition}\label{Prop:ST:ineq:quasi:mega}
For any Kripke model $\mathbb{M}$, any assignment $g$ on $\mathbb{M}$, and inequality $\mathsf{Ineq}$, quasi-inequality $\mathsf{Quasi}$, mega-inequality $\mathsf{Mega}$, universally quantified inequality $\mathsf{UQIneq}$, quasi-universally quantified inequality $\mathsf{QUQIneq}$,

$$\mathbb{M},g\Vdash\mathsf{Ineq}\mbox{ iff }\mathbb{M},g\vDash ST(\mathsf{Ineq});$$
$$\mathbb{M},g\Vdash\mathsf{Quasi}\mbox{ iff }\mathbb{M},g\vDash ST(\mathsf{Quasi});$$
$$\mathbb{M},g\Vdash\mathsf{Mega}\mbox{ iff }\mathbb{M},g\vDash ST(\mathsf{Mega});$$
$$\mathbb{M},g\Vdash\mathsf{UQIneq}\mbox{ iff }\mathbb{M},g\vDash ST(\mathsf{UQIneq});$$
$$\mathbb{M},g\Vdash\mathsf{QUQIneq}\mbox{ iff }\mathbb{M},g\vDash ST(\mathsf{QUQIneq}).$$
\end{proposition}

\section{Sahlqvist inequalities}\label{Sec:Sahl}

In the present section, since we will use the algorithm $\mathsf{ALBA}^{\downarrow}$ which is based on the classsification of nodes in the signed generation trees of hybrid modal formulas, we will use the unified correspondence style definition (cf.\ \cite{CPSZ,CPZ:Trans,PaSoZh16,Zh21}) to define Sahlqvist inequalities. We will collect all the necessary preliminaries on Sahlqvist inequalities. 

\begin{definition}[Order-type of propositional variables](cf.\ \cite[page 346]{CoPa12})
For an $n$-tuple $(p_1, \ldots, p_n)$ of propositional variables, an order-type $\epsilon$ of $(p_1, \ldots, p_n)$ is an element in $\{1,\partial\}^{n}$. We say that $p_i$ has order-type 1 with respect to $\epsilon$ if $\epsilon_i=1$, and denote $\epsilon(p_i)=1$ or $\epsilon(i)=1$; we say that $p_i$ has order-type $\partial$ with respect to $\epsilon$ if $\epsilon_i=\partial$, and denote $\epsilon(p_i)=\partial$ or $\epsilon(i)=\partial$. We use $\epsilon^{\partial}$ to denote the order-type such that $\epsilon^{\partial}(p_i)=1$ (resp.\ $\epsilon^{\partial}(p_i)=\partial$) iff $\epsilon(p_i)=\partial$ (resp.\ $\epsilon(p_i)=1$).
\end{definition}

\begin{definition}[Signed generation tree]\label{adef: signed gen tree}(cf.\ \cite[Definition 4]{CPZ:Trans})
The \emph{positive} (resp.\ \emph{negative}) {\em generation tree} of any given formula $\phi$ is defined by first labelling the root of the generation tree of $\phi$ with $+$ (resp.\ $-$) and then labelling the children nodes as follows:
\begin{itemize}
\item Assign the same sign to the children nodes of any node labelled with $\lor, \land, \Box$, $\Diamond$, $\downarrow x$, $\blacksquare, \Diamondblack,\mathsf{A},\mathsf{E},\forall x, \exists x$;
\item Assign the opposite sign to the child node of any node labelled with $\neg$;
\item Assign the opposite sign to the first child node and the same sign to the second child node of any node labelled with $\to$;
\item Assign the same sign to the second child node labelled with $@$ (notice that we do not label the first child node with nominal or state variable).
\end{itemize}
Nodes in signed generation trees are \emph{positive} (resp.\ \emph{negative}) if they are signed $+$ (resp.\ $-$).
\end{definition}

\begin{example}

The positive generation tree of $+\Box (p \lor \neg\Diamond q)\to\Box q$ is given in Figure \ref{fig:Church:Rosser}.

\begin{figure}[htb]
\centering
\begin{tikzpicture}
\tikzstyle{level 1}=[level distance=0.8cm, sibling distance=1.5cm]
\tikzstyle{level 2}=[level distance=0.8cm, sibling distance=1.5cm]
\tikzstyle{level 3}=[level distance=0.8cm, sibling distance=1.5cm]
 \node {$+\to$}         
              child{node{$-\Box$}
                     child{node{$-\lor$}
                           child{node{$-p$}}
                            child{node{$-\neg$}
                             child{node{$+\Diamond$}
                                child{node{$+q$}}}}}}
              child{node{$+\Box$}
                        child{node{$+q$}}} 
;
\end{tikzpicture}
\caption{Positive generation tree for $\Box (p \lor \neg\Diamond q)\to\Box q$}
\label{fig:Church:Rosser}
\end{figure}
\end{example}

Signed generation trees will be used in the inequalities $\phi\leq\psi$, where the positive generation tree $+\phi$ and the negative generation tree $-\psi$ will be considered. We will also say that an inequality $\phi\leq\psi$ is \emph{uniform} in a variable $p_i$ if all occurrences of $p_i$ in $+\phi$ and $-\psi$ have the same sign, and that $\phi\leq\psi$ is $\epsilon$-\emph{uniform} in an array $\vec{p}$ if $\phi\leq\psi$ is uniform in $p_i$, occurring with the sign indicated by $\epsilon$ (i.e., $p_i$ has the sign $+$ if $\epsilon(p_i)=1$, and has the sign $-$ if $\epsilon(p_i)=\partial$), for each propositional variable $p_i$ in $\vec{p}$.

For any given formula $\phi(p_1,\ldots p_n)$, any order-type $\epsilon$ over $n$, and any $1 \leq i \leq n$, an \emph{$\epsilon$-critical node} in a signed generation tree of $\phi$ is a leaf node $+p_i$ when $\epsilon_i = 1$ or $-p_i$ when $\epsilon_i = \partial$. An $\epsilon$-{\em critical branch} in a signed generation tree is a branch from an $\epsilon$-critical node. The $\epsilon$-critical occurrences are intended to be those which the algorithm $\mathsf{ALBA}^{\downarrow}$ will solve for. We say that $+\phi$ (resp.\ $-\phi$) {\em agrees with} $\epsilon$, and write $\epsilon(+\phi)$ (resp.\ $\epsilon(-\phi)$), if every leaf node with a propositional variable in the signed generation tree of $+\phi$ (resp.\ $-\phi$) is $\epsilon$-critical.

We will also use the notation $+\psi\prec \ast \phi$ (resp.\ $-\psi\prec \ast \phi$) to indicate that an occurrence (it does not matter which occurrence it is) of a subformula $\psi$ inherits the positive (resp.\ negative) sign from the signed generation tree $\ast \phi$, where $\ast\in\{+,-\}$. We will write $\epsilon(\gamma) \prec \ast \phi$ (resp.\ $\epsilon^\partial(\gamma) \prec \ast \phi$) to indicate that the signed generation subtree $\gamma$, with the sign inherited from $\ast \phi$, agrees with $\epsilon$ (resp.\ with $\epsilon^\partial$). We say that a propositional variable $p$ is \emph{positive} (resp.\ \emph{negative}) in $\phi$ if $+p\prec+\phi$ (resp.\ $-p\prec+\phi$) for all occurrences of $p$ in $\phi$.\label{page:epsilon:subtree}

\begin{definition}\label{adef:good:branches}(cf.\ \cite[Definition 5]{CPZ:Trans})
Nodes in signed generation trees are called \emph{outer nodes} and \emph{inner nodes}, according to Table \ref{aJoin:and:Meet:Friendly:Table}. For the names of outer nodes and inner nodes in the classification, see Example \ref{Example:Sahlqvist}.

A branch in a signed generation tree is called a \emph{excellent branch} if it is the concatenation of two paths $P_1$ and $P_2$, one of which might be of length $0$, such that $P_1$ is a path from the leaf consisting (apart from variable nodes) of inner nodes only, and $P_2$ consists (apart from variable nodes) of outer nodes only.
\begin{table}
\begin{center}
\begin{tabular}{| c | c |}
\hline
Outer &Inner\\
\hline
\begin{tabular}{c c c c c c c c c c }
$+$ & $\vee$ & $\wedge$ &$\Diamond$ & $\neg$ & $\downarrow x$ & @\\
$-$ & $\wedge$ & $\vee$ &$\Box$ & $\neg$ & $\downarrow x$ & @ & $\to$\\
\end{tabular}
&
\begin{tabular}{c c c c c c}
$+$ &$\wedge$ &$\Box$ & $\neg$ & $\downarrow x$ & @\\
$-$ &$\vee$ &$\Diamond$ & $\neg$ & $\downarrow x$ & @\\
\end{tabular}\\
\hline
\end{tabular}
\end{center}
\caption{Outer and Inner nodes.}\label{aJoin:and:Meet:Friendly:Table}
\vspace{-1em}
\end{table}
\end{definition}

\begin{definition}[Sahlqvist inequalities]\label{aInducive:Ineq:Def}(cf.\ \cite[Definition 6]{CPZ:Trans})
For any order-type $\epsilon$, the signed generation tree $*\phi$ (where $*\in\{+,-\}$) of a formula $\phi(p_1,\ldots p_n)$ is \emph{$\epsilon$-Sahlqvist} if for all $1 \leq i \leq n$, every $\epsilon$-critical branch with leaf $p_i$ is excellent. An inequality $\phi\leq\psi$ is \emph{$\epsilon$-Sahlqvist} if the signed generation trees $+\phi$ and $-\psi$ are $\epsilon$-Sahlqvist. An inequality $\phi\leq\psi$ is \emph{Sahlqvist} if it is \emph{$\epsilon$}-Sahlqvist for some $\epsilon$.
\end{definition}

The classification of outer nodes and inner nodes is based on how different connectives behave in the algorithm. When the input inequality is a Sahlqvist inequality, the algorithm first decompose the outer part of the formula, and then decompose the inner part of the formula, which will be shown in Example \ref{Example:ALBA:execution}.

\begin{example}\label{Example:Sahlqvist}
Here we give an example of a Sahlqvist inequality for the order-type $\epsilon=(1,1)$, where the outer nodes are \textcolor{red}{red} and the inner nodes are \textcolor{green}{green}, and $\epsilon$-critical branches are ended with \textcolor{blue}{blue} leaf nodes. It is clear that the branch from $+p_2$ to $+\downarrow x$ and the branch from $+p_1$ to $+\downarrow x$ are both $\epsilon$-critical and excellent. This example also shows why the outer nodes and inner nodes are named as they are.

\begin{figure}[htb]
\centering
\begin{multicols}{3}
\begin{tikzpicture}
\tikzstyle{level 1}=[level distance=0.8cm, sibling distance=2cm]
\tikzstyle{level 2}=[level distance=0.8cm, sibling distance=2cm]
\tikzstyle{level 3}=[level distance=0.8cm, sibling distance=1.5cm]
 \node {\textcolor{red}{$+\downarrow x$}} 
              child{node{\textcolor{red}{$+\land$}}
              child{node{\textcolor{red}{$+\Diamond$}}
                     child{node{\textcolor{green}{$+\Box$}}
                           child{node{\textcolor{blue}{$+p_1$}}}}}
              child{node{\textcolor{green}{$+\Box$}}
                        child{node{\textcolor{green}{$+\land$}}
                        child{node{\textcolor{green}{$+@$}}
                        child{node{$\nomi$}}
                        child{node{$+x$}}}
                        child{node{\textcolor{green}{$+\Box$}}
                        child{node{\textcolor{blue}{$+p_2$}}}}}}} 
;
\end{tikzpicture}

\columnbreak

$\leq$

\columnbreak

\begin{tikzpicture}
\tikzstyle{level 1}=[level distance=1cm, sibling distance=1cm]
\tikzstyle{level 2}=[level distance=1cm, sibling distance=1cm]
\tikzstyle{level 3}=[level distance=1cm, sibling distance=1cm]
 \node {\textcolor{red}{$-\lor$}}         
                     child{node{\textcolor{green}{$-\Diamond$}}
                           child{node{$-\Box$}
                                 child{node{$-\Diamond$}
                                       child{node{$-p_1$}}}}}
                     child{node{\textcolor{green}{$-\Diamond$}}
                           child{node{$-\Box$}
                                 child{node{$-\Diamond$}
                                       child{node{$-p_2$}}}}}
;
\end{tikzpicture}
\end{multicols}
\caption{(1,1)-Sahlqvist inequality $\downarrow x.(\Diamond\Box p_1\land\Box(@_\nomi x\land\Box p_2))\leq \Diamond\Box\Diamond p_1\lor\Diamond\Box\Diamond p_2$}
\label{fig:Sahlqvist}
\end{figure}
\end{example}

As we can see from the example, the major difference of Sahlqvist formula here from the ones in \cite{ConRob} and \cite{tCMaVi06} is mainly that here we have downarrow binders as a part of the Sahlqvist structure.

\section{The algorithm $\mathsf{ALBA}^{\downarrow}$}\label{Sec:ALBA}

In the present section, we define the correspondence algorithm $\mathsf{ALBA}^{\downarrow}$ for hybrid logic with binder, in the style of \cite{CoGoVa06,CoPa12,Zh21}. The algorithm goes in three steps. 

\begin{enumerate}

\item \textbf{Preprocessing and first approximation}:

In the generation tree of $+\phi$ and $-\psi$\footnote{The discussion below relies on the definition of signed generation tree in Section \ref{Sec:Sahl}. In what follows, we identify a formula with its signed generation tree.},

\begin{enumerate}
\item Apply the distribution rules:

\begin{enumerate}
\item Push down $+\Diamond, -\neg, +\land, +\downarrow x, +@_{\nomi}, +@_{x}, -\to$ by distributing them over nodes labelled with $+\lor$ which are outer nodes (see Figure \ref{Figure:distribution:rules}; notice that here we treat $@_{\nomi}$ and $@_{x}$ as if they are unary modality with only the right branch as the input), and

\item Push down $-\Box,+\neg, -\lor, -\downarrow x, -@_{\nomi}, -@_{x}, -\to$ by distributing them over nodes labelled with $-\land$ which are outer nodes (see Figure \ref{Figure:distribution:rules:2}).

\end{enumerate}

\begin{figure}[htb]
\centering
\begin{multicols}{8}
\begin{tikzpicture}[scale=0.7]
\tikzstyle{level 1}=[level distance=1cm, sibling distance=1cm]
\tikzstyle{level 2}=[level distance=1cm, sibling distance=1cm]
\tikzstyle{level 3}=[level distance=1cm, sibling distance=1cm]
 \node {$+\Diamond$}         
              child{node{$+\lor$}
                     child{node{$+\alpha$}}
                           child{node{$+\beta$}}}
;
\end{tikzpicture}
\columnbreak

$\Rightarrow$
\columnbreak

\begin{tikzpicture}[scale=0.7]
\tikzstyle{level 1}=[level distance=1cm, sibling distance=1cm]
\tikzstyle{level 2}=[level distance=1cm, sibling distance=1cm]
\tikzstyle{level 3}=[level distance=1cm, sibling distance=1cm]
 \node {$+\lor$}
              child{node{$+\Diamond$}
                     child{node{$+\alpha$}}}
              child{node{$+\Diamond$}
                           child{node{$+\beta$}}}
 ;
\end{tikzpicture}

\columnbreak

$\ $
\columnbreak

$\ $
\columnbreak

\begin{tikzpicture}[scale=0.7]
\tikzstyle{level 1}=[level distance=1cm, sibling distance=1cm]
\tikzstyle{level 2}=[level distance=1cm, sibling distance=1cm]
\tikzstyle{level 3}=[level distance=1cm, sibling distance=1cm]
 \node {$-\neg$}         
              child{node{$+\lor$}
                     child{node{$+\alpha$}}
                           child{node{$+\beta$}}}
 ;
\end{tikzpicture}

\columnbreak

$\Rightarrow$
\columnbreak

\begin{tikzpicture}[scale=0.7]
\tikzstyle{level 1}=[level distance=1cm, sibling distance=1cm]
\tikzstyle{level 2}=[level distance=1cm, sibling distance=1cm]
\tikzstyle{level 3}=[level distance=1cm, sibling distance=1cm]
 \node {$-\land$}
              child{node{$-\neg$}
                     child{node{$+\alpha$}}}
              child{node{$-\neg$}
                           child{node{$+\beta$}}}
 ;
\end{tikzpicture}
\end{multicols}

\centering
\begin{multicols}{8}
\begin{tikzpicture}[scale=0.7]
\tikzstyle{level 1}=[level distance=1cm, sibling distance=1cm]
\tikzstyle{level 2}=[level distance=1cm, sibling distance=1cm]
\tikzstyle{level 3}=[level distance=1cm, sibling distance=1cm]
 \node {$+\land$}         
              child{node{$+\lor$}
                     child{node{$+\alpha$}}
                           child{node{$+\beta$}}}
              child{node{+$\gamma$}}
;
\end{tikzpicture}
\columnbreak

$\ \ \ \ \ \ \ \ \Rightarrow$
\columnbreak

\begin{tikzpicture}[scale=0.7]
\tikzstyle{level 1}=[level distance=1cm, sibling distance=2cm]
\tikzstyle{level 2}=[level distance=1cm, sibling distance=1cm]
\tikzstyle{level 3}=[level distance=1cm, sibling distance=1cm]
 \node {$+\lor$}
              child{node{$+\land$}
                     child{node{$+\alpha$}}
                     child{node{$+\gamma$}}}
              child{node{$+\land$}
                           child{node{$+\beta$}}
                           child{node{$+\gamma$}}}
 ;
\end{tikzpicture}

\columnbreak

$\ $
\columnbreak

$\ $
\columnbreak

\begin{tikzpicture}[scale=0.7]
\tikzstyle{level 1}=[level distance=1cm, sibling distance=1cm]
\tikzstyle{level 2}=[level distance=1cm, sibling distance=1cm]
\tikzstyle{level 3}=[level distance=1cm, sibling distance=1cm]
 \node {$+\land$}         
              child{node{+$\alpha$}}
              child{node{$+\lor$}
                     child{node{$+\beta$}}
                           child{node{$+\gamma$}}}
 ;
\end{tikzpicture}

\columnbreak

$\ \ \ \ \ \ \ \ \Rightarrow$
\columnbreak

\begin{tikzpicture}[scale=0.7]
\tikzstyle{level 1}=[level distance=1cm, sibling distance=2cm]
\tikzstyle{level 2}=[level distance=1cm, sibling distance=1cm]
\tikzstyle{level 3}=[level distance=1cm, sibling distance=1cm]
 \node {$+\lor$}
              child{node{$+\land$}
                     child{node{$+\alpha$}}
                     child{node{$+\beta$}}}
              child{node{$+\land$}
                           child{node{$+\alpha$}}
                           child{node{$+\gamma$}}}
 ;
\end{tikzpicture}
\end{multicols}

\centering
\begin{multicols}{8}
\begin{tikzpicture}[scale=0.7]
\tikzstyle{level 1}=[level distance=1cm, sibling distance=1cm]
\tikzstyle{level 2}=[level distance=1cm, sibling distance=1cm]
\tikzstyle{level 3}=[level distance=1cm, sibling distance=1cm]
 \node {$+\downarrow x$}         
              child{node{$+\lor$}
                     child{node{$+\alpha$}}
                           child{node{$+\beta$}}}
;
\end{tikzpicture}
\columnbreak

$\Rightarrow$
\columnbreak

\begin{tikzpicture}[scale=0.7]
\tikzstyle{level 1}=[level distance=1cm, sibling distance=1.5cm]
\tikzstyle{level 2}=[level distance=1cm, sibling distance=1.5cm]
\tikzstyle{level 3}=[level distance=1cm, sibling distance=1.5cm]
 \node {$+\lor$}
              child{node{$+\downarrow x$}
                     child{node{$+\alpha$}}}
              child{node{$+\downarrow x$}
                           child{node{$+\beta$}}}
 ;
\end{tikzpicture}

\columnbreak

$\ $
\columnbreak

$\ $
\columnbreak

\begin{tikzpicture}[scale=0.7]
\tikzstyle{level 1}=[level distance=1cm, sibling distance=1cm]
\tikzstyle{level 2}=[level distance=1cm, sibling distance=1cm]
\tikzstyle{level 3}=[level distance=1cm, sibling distance=1cm]
 \node {$+@_{\nomi}$}         
              child{node{$+\lor$}
                     child{node{$+\alpha$}}
                           child{node{$+\beta$}}}
;
\end{tikzpicture}
\columnbreak

$\Rightarrow$
\columnbreak

\begin{tikzpicture}[scale=0.7]
\tikzstyle{level 1}=[level distance=1cm, sibling distance=1cm]
\tikzstyle{level 2}=[level distance=1cm, sibling distance=1cm]
\tikzstyle{level 3}=[level distance=1cm, sibling distance=1cm]
 \node {$+\lor$}
              child{node{$+@_{\nomi}$}
                     child{node{$+\alpha$}}}
              child{node{$+@_{\nomi}$}
                           child{node{$+\beta$}}}
 ;
\end{tikzpicture}
\end{multicols}

\centering
\begin{multicols}{8}
\begin{tikzpicture}[scale=0.7]
\tikzstyle{level 1}=[level distance=1cm, sibling distance=1cm]
\tikzstyle{level 2}=[level distance=1cm, sibling distance=1cm]
\tikzstyle{level 3}=[level distance=1cm, sibling distance=1cm]
 \node {$+@_{x}$}         
              child{node{$+\lor$}
                     child{node{$+\alpha$}}
                           child{node{$+\beta$}}}
;
\end{tikzpicture}
\columnbreak

$\Rightarrow$
\columnbreak

\begin{tikzpicture}[scale=0.7]
\tikzstyle{level 1}=[level distance=1cm, sibling distance=1cm]
\tikzstyle{level 2}=[level distance=1cm, sibling distance=1cm]
\tikzstyle{level 3}=[level distance=1cm, sibling distance=1cm]
 \node {$+\lor$}
              child{node{$+@_{x}$}
                     child{node{$+\alpha$}}}
              child{node{$+@_{x}$}
                           child{node{$+\beta$}}}
 ;
\end{tikzpicture}

\columnbreak

$\ $
\columnbreak

$\ $
\columnbreak

\begin{tikzpicture}[scale=0.7]
\tikzstyle{level 1}=[level distance=1cm, sibling distance=1cm]
\tikzstyle{level 2}=[level distance=1cm, sibling distance=1cm]
\tikzstyle{level 3}=[level distance=1cm, sibling distance=1cm]
 \node {$-\to$}         
              child{node{$+\lor$}
                     child{node{$+\alpha$}}
                           child{node{$+\beta$}}}
              child{node{$-\gamma$}}
;
\end{tikzpicture}
\columnbreak

$\ \ \ \ \ \ \ \ \Rightarrow$
\columnbreak

\begin{tikzpicture}[scale=0.7]
\tikzstyle{level 1}=[level distance=1cm, sibling distance=2cm]
\tikzstyle{level 2}=[level distance=1cm, sibling distance=1cm]
\tikzstyle{level 3}=[level distance=1cm, sibling distance=1cm]
 \node {$-\land$}
              child{node{$-\to$}
                     child{node{$+\alpha$}}
                     child{node{$-\gamma$}}}
              child{node{$-\to$}
                           child{node{$+\beta$}}
                           child{node{$-\gamma$}}}
 ;
\end{tikzpicture}
\end{multicols}
\caption{Distribution rules for $+\lor$}
\label{Figure:distribution:rules}
\end{figure}

\begin{figure}[htb]

\centering
\begin{multicols}{8}
\begin{tikzpicture}[scale=0.7]
\tikzstyle{level 1}=[level distance=1cm, sibling distance=1cm]
\tikzstyle{level 2}=[level distance=1cm, sibling distance=1cm]
\tikzstyle{level 3}=[level distance=1cm, sibling distance=1cm]
 \node {$-\Box$}         
              child{node{$-\land$}
                     child{node{$-\alpha$}}
                           child{node{$-\beta$}}}
 ;
\end{tikzpicture}

\columnbreak

$\ \ \ \ \ \ \ \ \Rightarrow$
\columnbreak

\begin{tikzpicture}[scale=0.7]
\tikzstyle{level 1}=[level distance=1cm, sibling distance=1cm]
\tikzstyle{level 2}=[level distance=1cm, sibling distance=1cm]
\tikzstyle{level 3}=[level distance=1cm, sibling distance=1cm]
 \node {$-\land$}
              child{node{$-\Box$}
                     child{node{$-\alpha$}}}
              child{node{$-\Box$}
                           child{node{$-\beta$}}}
 ;
\end{tikzpicture}

\columnbreak

$\ $
\columnbreak

$\ $
\columnbreak

\begin{tikzpicture}[scale=0.7]
\tikzstyle{level 1}=[level distance=1cm, sibling distance=1cm]
\tikzstyle{level 2}=[level distance=1cm, sibling distance=1cm]
\tikzstyle{level 3}=[level distance=1cm, sibling distance=1cm]
 \node {$+\neg$}         
              child{node{$-\land$}
                     child{node{$-\alpha$}}
                           child{node{$-\beta$}}}
;
\end{tikzpicture}
\columnbreak

$\Rightarrow$
\columnbreak

\begin{tikzpicture}[scale=0.7]
\tikzstyle{level 1}=[level distance=1cm, sibling distance=1cm]
\tikzstyle{level 2}=[level distance=1cm, sibling distance=1cm]
\tikzstyle{level 3}=[level distance=1cm, sibling distance=1cm]
 \node {$+\lor$}
              child{node{$+\neg$}
                     child{node{$-\alpha$}}}
              child{node{$+\neg$}
                           child{node{$-\beta$}}}
 ;
\end{tikzpicture}
\end{multicols}

\centering
\begin{multicols}{8}
\begin{tikzpicture}[scale=0.7]
\tikzstyle{level 1}=[level distance=1cm, sibling distance=1cm]
\tikzstyle{level 2}=[level distance=1cm, sibling distance=1cm]
\tikzstyle{level 3}=[level distance=1cm, sibling distance=1cm]
 \node {$-\lor$}         
              child{node{$-\land$}
                     child{node{$-\alpha$}}
                           child{node{$-\beta$}}}
              child{node{$-\gamma$}}
 ;
\end{tikzpicture}

\columnbreak

$\Rightarrow$
\columnbreak

\begin{tikzpicture}[scale=0.7]
\tikzstyle{level 1}=[level distance=1cm, sibling distance=2cm]
\tikzstyle{level 2}=[level distance=1cm, sibling distance=1cm]
\tikzstyle{level 3}=[level distance=1cm, sibling distance=1cm]
 \node {$-\land$}
              child{node{$-\lor$}
                     child{node{$-\alpha$}}
                     child{node{$-\gamma$}}}
              child{node{$-\lor$}
                           child{node{$-\beta$}}
                           child{node{$-\gamma$}}}
 ;
\end{tikzpicture}

\columnbreak

$\ $
\columnbreak

$\ $
\columnbreak

\begin{tikzpicture}[scale=0.7]
\tikzstyle{level 1}=[level distance=1cm, sibling distance=1cm]
\tikzstyle{level 2}=[level distance=1cm, sibling distance=1cm]
\tikzstyle{level 3}=[level distance=1cm, sibling distance=1cm]
 \node {$-\lor$}         
              child{node{$-\alpha$}}
              child{node{$-\land$}
                     child{node{$-\beta$}}
                           child{node{$-\gamma$}}}
;
\end{tikzpicture}
\columnbreak

$\Rightarrow$
\columnbreak

\begin{tikzpicture}[scale=0.7]
\tikzstyle{level 1}=[level distance=1cm, sibling distance=2cm]
\tikzstyle{level 2}=[level distance=1cm, sibling distance=1cm]
\tikzstyle{level 3}=[level distance=1cm, sibling distance=1cm]
 \node {$-\land$}
              child{node{$-\lor$}
                     child{node{$-\alpha$}}
                     child{node{$-\beta$}}}
              child{node{$-\lor$}
                           child{node{$-\alpha$}}
                           child{node{$-\gamma$}}}
 ;
\end{tikzpicture}
\end{multicols}

\centering
\begin{multicols}{8}
\begin{tikzpicture}[scale=0.7]
\tikzstyle{level 1}=[level distance=1cm, sibling distance=1cm]
\tikzstyle{level 2}=[level distance=1cm, sibling distance=1cm]
\tikzstyle{level 3}=[level distance=1cm, sibling distance=1cm]
 \node {$-\downarrow x$}         
              child{node{$-\land$}
                     child{node{$-\alpha$}}
                           child{node{$-\beta$}}}
 ;
\end{tikzpicture}

\columnbreak

$\ \ \ \ \ \ \ \ \Rightarrow$
\columnbreak

\begin{tikzpicture}[scale=0.7]
\tikzstyle{level 1}=[level distance=1cm, sibling distance=1.5cm]
\tikzstyle{level 2}=[level distance=1cm, sibling distance=1.5cm]
\tikzstyle{level 3}=[level distance=1cm, sibling distance=1.5cm]
 \node {$-\land$}
              child{node{$-\downarrow x$}
                     child{node{$-\alpha$}}}
              child{node{$-\downarrow x$}
                           child{node{$-\beta$}}}
 ;
\end{tikzpicture}

\columnbreak

$\ $
\columnbreak

$\ $
\columnbreak

\begin{tikzpicture}[scale=0.7]
\tikzstyle{level 1}=[level distance=1cm, sibling distance=1cm]
\tikzstyle{level 2}=[level distance=1cm, sibling distance=1cm]
\tikzstyle{level 3}=[level distance=1cm, sibling distance=1cm]
 \node {$-@_{\nomi}$}         
              child{node{$-\land$}
                     child{node{$-\alpha$}}
                           child{node{$-\beta$}}}
 ;
\end{tikzpicture}

\columnbreak

$\ \ \ \ \ \ \ \ \Rightarrow$
\columnbreak

\begin{tikzpicture}[scale=0.7]
\tikzstyle{level 1}=[level distance=1cm, sibling distance=1cm]
\tikzstyle{level 2}=[level distance=1cm, sibling distance=1cm]
\tikzstyle{level 3}=[level distance=1cm, sibling distance=1cm]
 \node {$-\land$}
              child{node{$-@_{\nomi}$}
                     child{node{$-\alpha$}}}
              child{node{$-@_{\nomi}$}
                           child{node{$-\beta$}}}
 ;
\end{tikzpicture}\end{multicols}

\centering
\begin{multicols}{8}
\begin{tikzpicture}[scale=0.7]
\tikzstyle{level 1}=[level distance=1cm, sibling distance=1cm]
\tikzstyle{level 2}=[level distance=1cm, sibling distance=1cm]
\tikzstyle{level 3}=[level distance=1cm, sibling distance=1cm]
 \node {$-@_{x}$}         
              child{node{$-\land$}
                     child{node{$-\alpha$}}
                           child{node{$-\beta$}}}
 ;
\end{tikzpicture}

\columnbreak

$\ \ \ \ \ \ \ \ \Rightarrow$
\columnbreak

\begin{tikzpicture}[scale=0.7]
\tikzstyle{level 1}=[level distance=1cm, sibling distance=1cm]
\tikzstyle{level 2}=[level distance=1cm, sibling distance=1cm]
\tikzstyle{level 3}=[level distance=1cm, sibling distance=1cm]
 \node {$-\land$}
              child{node{$-@_{x}$}
                     child{node{$-\alpha$}}}
              child{node{$-@_{x}$}
                           child{node{$-\beta$}}}
 ;
\end{tikzpicture}

\columnbreak

$\ $
\columnbreak

$\ $
\columnbreak

\begin{tikzpicture}[scale=0.7]
\tikzstyle{level 1}=[level distance=1cm, sibling distance=1cm]
\tikzstyle{level 2}=[level distance=1cm, sibling distance=1cm]
\tikzstyle{level 3}=[level distance=1cm, sibling distance=1cm]
 \node {$-\to$}         
              child{node{+$\alpha$}}
              child{node{$-\land$}
                     child{node{$-\beta$}}
                           child{node{$-\gamma$}}}
 ;
\end{tikzpicture}

\columnbreak

$\Rightarrow$
\columnbreak

\begin{tikzpicture}[scale=0.7]
\tikzstyle{level 1}=[level distance=1cm, sibling distance=2cm]
\tikzstyle{level 2}=[level distance=1cm, sibling distance=1cm]
\tikzstyle{level 3}=[level distance=1cm, sibling distance=1cm]
 \node {$-\land$}
              child{node{$-\to$}
                     child{node{$+\alpha$}}
                     child{node{$-\beta$}}}
              child{node{$-\to$}
                           child{node{$+\alpha$}}
                           child{node{$-\gamma$}}}
 ;
\end{tikzpicture}
\end{multicols}
\caption{Distribution rules for $-\land$}
\label{Figure:distribution:rules:2}
\end{figure}

\item Apply the splitting rules:

$$\infer{\alpha\leq\beta\ \ \ \alpha\leq\gamma}{\alpha\leq\beta\land\gamma}
\qquad
\infer{\alpha\leq\gamma\ \ \ \beta\leq\gamma}{\alpha\lor\beta\leq\gamma}
$$

\item Apply the monotone and antitone variable-elimination rules\footnote{Here the monotone and antitone variable elimination rules eliminate propositional variables $p$ such that the inequality is semantically monotone or antitone with respect to $p$.}:

$$\infer{\alpha(\perp)\leq\beta(\perp)}{\alpha(p)\leq\beta(p)}
\qquad
\infer{\beta(\top)\leq\alpha(\top)}{\beta(p)\leq\alpha(p)}
$$

for $\beta(p)$ positive in $p$ and $\alpha(p)$ negative in $p$.

\end{enumerate}

We denote by $\mathsf{Preprocess}(\phi\leq\psi)$ the finite set $\{\phi_i\leq\psi_i\}_{i\in I}$ of inequalities obtained after the exhaustive application of the previous rules. Then we apply the following first approximation rule to every inequality in $\mathsf{Preprocess}(\phi\leq\psi)$:

$$\infer{\nomi_0\leq\phi_i\ \ \ \psi_i\leq \neg\nomi_1}{\phi_i\leq\psi_i}
$$

Here, $\nomi_0$ and $\nomi_1$ are special fresh nominals. Now we get a set of inequalities $\{\nomi_0\leq\phi_i, \psi_i\leq \neg\nomi_1\}_{i\in I}$. We call the set $\{\nomi_0\leq\phi_i, \psi_i\leq \neg\nomi_1\}$ a \emph{system}.

\item \textbf{The reduction stage}:

In this stage, for each $\{\nomi_0\leq\phi_i, \psi_i\leq \neg\nomi_1\}$, we apply the following rules to prepare for eliminating all the proposition variables in $\{\nomi_0\leq\phi_i, \psi_i\leq\neg\nomi_1\}$:

\begin{enumerate}

\item \textbf{Substage 1: Decomposing the outer part}

In the current substage, the following rules are applied to decompose the outer part of the Sahlqvist signed formula:

\begin{enumerate}

\item Splitting rules:

$$
\infer{\alpha\leq\beta\ \ \ \alpha\leq\gamma}{\alpha\leq\beta\land\gamma}
\qquad
\infer{\alpha\leq\gamma\ \ \ \beta\leq\gamma}{\alpha\lor\beta\leq\gamma}
$$

\item Approximation rules:
$$
\infer{\nomj\leq\alpha\ \ \ \nomi\leq\Diamond\nomj}{\nomi\leq\Diamond\alpha}
\qquad
\infer{\nomj\leq\alpha\ \ \ x\leq\Diamond\nomj}{x\leq\Diamond\alpha}
$$

$$
\infer{\alpha\leq\neg\nomj\ \ \ \Box\neg\nomj\leq\neg\nomi}{\Box\alpha\leq\neg\nomi}
\qquad
\infer{\alpha\leq\neg\nomj\ \ \ \Box\neg\nomj\leq\neg x}{\Box\alpha\leq\neg x}
$$

$$
\infer{\nomj\leq\alpha}{\nomi\leq @_{\nomj}\alpha}
\qquad
\infer{\nomj\leq\alpha}{x\leq @_{\nomj}\alpha}
$$

$$
\infer{\alpha\leq\neg\nomj}{@_{\nomj}\alpha\leq\neg\nomi}
\qquad
\infer{\alpha\leq\neg\nomj}{@_{\nomj}\alpha\leq\neg x}
$$

$$
\infer{x\leq\alpha}{\nomi\leq @_{x}\alpha}
\qquad
\infer{x\leq\alpha}{y\leq @_{x}\alpha}
$$

$$
\infer{\alpha\leq\neg x}{@_{x}\alpha\leq\neg\nomi}
\qquad
\infer{\alpha\leq\neg x}{@_{x}\alpha\leq\neg y}
$$

$$
\infer{\nomi\leq\alpha[\nomi/x]}{\nomi\leq \downarrow x.\alpha}
\qquad
\infer{y\leq\alpha[y/x]}{y\leq \downarrow x.\alpha}
$$

$$
\infer{\alpha[\nomi/x]\leq\neg\nomi}{\downarrow x.\alpha\leq\neg\nomi}
\qquad
\infer{\alpha[y/x]\leq\neg y}{\downarrow x.\alpha\leq\neg y}
$$

$$
\infer{\nomj\leq\alpha\ \ \ \ \ \ \ \beta\leq\neg\nomk\ \ \ \ \ \ \ \nomj\rightarrow\neg\nomk\leq\neg\nomi}{\alpha\rightarrow\beta\leq\neg\nomi}
$$

$$\infer{\nomj\leq\alpha\ \ \ \ \ \ \ \beta\leq\neg\nomk\ \ \ \ \ \ \ \nomj\rightarrow\neg\nomk\leq\neg x}{\alpha\rightarrow\beta\leq\neg x}
$$

The nominals introduced by the approximation rules must not occur in the system before applying the rule, and $\alpha[\nomi/x]$ (resp.\ $\alpha[y/x]$) indicates that all occurrences of $x$ in $\alpha$ are replaced by $\nomi$ (resp.\ $y$).

\item Residuation rules:

$$
\infer{\alpha\leq\neg\nomi}{\nomi\leq\neg\alpha}
\qquad
\infer{\nomi\leq\alpha}{\neg\alpha\leq\neg\nomi}
\qquad
\infer{\alpha\leq\neg x}{x\leq\neg\alpha}
\qquad
\infer{x\leq\alpha}{\neg\alpha\leq\neg x}
$$

\end{enumerate}

\item \textbf{Substage 2: Decomposing the inner part}

In the current substage, the following rules are applied to decompose the inner part of the Sahlqvist signed formula. We allow these rules to be applied in the scope of universal quantifiers, as we will show in Example \ref{Example:ALBA:execution:2}:

\begin{enumerate}

\item Splitting rules:

$$\infer{\alpha\leq\beta\ \ \ \alpha\leq\gamma}{\alpha\leq\beta\land\gamma}
\qquad
\infer{\alpha\leq\gamma\ \ \ \beta\leq\gamma}{\alpha\lor\beta\leq\gamma}
$$

\item Residuation rules:

$$\infer{\beta\leq\neg\alpha}{\alpha\leq\neg\beta}
\qquad
\infer{\neg\beta\leq\alpha}{\neg\alpha\leq\beta}
\qquad
\infer{\alpha\leq\blacksquare\beta}{\Diamond\alpha\leq\beta}
\qquad
\infer{\Diamondblack\alpha\leq\beta}{\alpha\leq\Box\beta}
$$

$$\infer{\mathsf{E}\alpha\land\nomj\leq\beta}{\alpha\leq @_{\nomj}\beta}
\qquad
\infer{\beta\leq\nomj\to\mathsf{A}\alpha}{@_{\nomj}\beta\leq\alpha}
\qquad
\infer{\mathsf{E}\alpha\land x\leq\beta}{\alpha\leq @_{x}\beta}
\qquad
\infer{\beta\leq x\to\mathsf{A}\alpha}{@_{x}\beta\leq\alpha}
$$

$$\infer{\forall y(\mathsf{A}(y\to\alpha)\land y\leq\beta[y/x])}{\alpha\leq\downarrow x.\beta}
\qquad
\infer{\forall y(\beta[y/x]\leq y\to\mathsf{E}(y\land\alpha))}{\downarrow x.\beta\leq\alpha}
$$

The state variables introduced by the residuation rules must not occur in the system before applying the rule.

\item Second splitting rule:
$$\infer{\forall x(\mathsf{Mega_1})\ \ \ \forall x(\mathsf{Mega_2})}
{\forall x(\mathsf{Mega_1} \bigamp\mathsf{Mega_2})}$$

Here $\mathsf{Mega_1}$ and $\mathsf{Mega_2}$ denote mega-inequalities.

\end{enumerate}

\item \textbf{Substage 3: Preparing for the Ackermann rules}

In this substage, we prepare for eliminating the propositional variables by the Ackermann rules, with the help of the following packing rules. Here we also allow this rule to be applied inside the scope of universal quantifiers.\\

Packing rules:

$$\infer{(\exists x\alpha)\leq\beta}{\forall x(\alpha\leq\beta)}
\qquad
\infer{\beta\leq(\forall x\alpha)}{\forall x(\beta\leq\alpha)}
$$

where $\beta$ does not contain occurrences of $x$.

\item \textbf{Substage 4: The Ackermann stage}

In this substage, we compute the minimal/maximal valuation for propositional variables and use the Ackermann rules to eliminate all the propositional variables.  These two rules are the core of $\mathsf{ALBA}$, since their application eliminates proposition variables. In fact, all the preceding steps are aimed at reaching a shape in which the rules can be applied. Notice that an important feature of these rules is that they are executed on the whole set of (universally quantified) inequalities, and not on a single inequality.\\

The right-handed Ackermann rule:

The system 
$\left\{ \begin{array}{ll}
\alpha_1\leq p \\
\vdots\\
\alpha_n\leq p \\
\forall{\vec x_1}(\beta_1\leq\gamma_1)\\
\vdots\\
\forall{\vec x_m}(\beta_m\leq\gamma_m)\\

\end{array} \right.$ 

is replaced by 
$\left\{ \begin{array}{ll}
\forall{\vec x_1}(\beta_1[(\alpha_1\lor\ldots\lor\alpha_n)/p]\leq\gamma_1[(\alpha_1\lor\ldots\lor\alpha_n)/p]) \\
\vdots\\
\forall{\vec x_m}(\beta_m[(\alpha_1\lor\ldots\lor\alpha_n)/p]\leq\gamma_m[(\alpha_1\lor\ldots\lor\alpha_n)/p]) \\

\end{array} \right.$

where:

\begin{enumerate}
\item $p, {\vec x_1}, \ldots, {\vec x_m}$ do not occur in $\alpha_1, \ldots, \alpha_n$;
\item Each $\beta_i$ is positive, and each $\gamma_i$ negative in $p$, for $1\leq i\leq m$;
\item Each $\alpha_i$ is pure.
\end{enumerate}

The left-handed Ackermann rule:

The system
$\left\{ \begin{array}{ll}
p\leq\alpha_1 \\
\vdots\\
p\leq\alpha_n \\
\forall{\vec x_1}(\beta_1\leq\gamma_1)\\
\vdots\\
\forall{\vec x_m}(\beta_m\leq\gamma_m)\\

\end{array} \right.$

is replaced by
$\left\{ \begin{array}{ll}
\forall{\vec x_1}(\beta_1[(\alpha_1\lor\ldots\lor\alpha_n)/p]\leq\gamma_1[(\alpha_1\lor\ldots\lor\alpha_n)/p]) \\
\vdots\\
\forall{\vec x_m}(\beta_m[(\alpha_1\lor\ldots\lor\alpha_n)/p]\leq\gamma_m[(\alpha_1\lor\ldots\lor\alpha_n)/p]) \\

\end{array} \right.$

where:
\begin{enumerate}
\item $p, {\vec x_1}, \ldots, {\vec x_m}$ do not occur in $\alpha_1, \ldots, \alpha_n$;
\item Each $\beta_i$ is negative, and each $\gamma_i$ positive in $p$, for $1\leq i\leq m$.
\item Each $\alpha_i$ is pure.
\end{enumerate}
\end{enumerate}

\item \textbf{Output}: If in the previous stage, for some $\{\nomi_0\leq\phi_i, \psi_i\leq \neg\nomi_1\}$, the algorithm gets stuck, i.e.\ some proposition variables cannot be eliminated by the application of the reduction rules, then the algorithm halts and output ``failure''. Otherwise, each initial tuple $\{\nomi_0\leq\phi_i, \psi_i\leq \neg\nomi_1\}$ of inequalities after the first approximation has been reduced to a set of pure (universally quantified) inequalities $\mathsf{Reduce}(\phi_i\leq\psi_i)$, and then the output is a set of quasi-(universally quantified) inequalities $\{\&\mathsf{Reduce}(\phi_i\leq\psi_i)\Rightarrow \nomi_0\leq \neg\nomi_1: \phi_i\leq\psi_i\in\mathsf{Preprocess}(\phi\leq\psi)\}$. Then the algorithm use the standard translation to transform the quasi-(universally quantified) inequalities into first-order formulas. Finally, use universal quantifiers to quantify all free individual variables $x$ and individual constants $i$ in the standard translation.
\end{enumerate}

Here we give an example of the execution of $\mathsf{ALBA}^{\downarrow}$ on the for the $(1,1)$-Sahlqvist inequality $\downarrow x.(\Diamond\Box p_1\land\Box(@_\nomi x\land\Box p_2))\leq \Diamond\Box\Diamond p_1\lor\Diamond\Box\Diamond p_2$.

\begin{example}\label{Example:ALBA:execution}

The algorithm receives the input inequality $\downarrow x.(\Diamond\Box p_1\land\Box(@_\nomi x\land\Box p_2))\leq \Diamond\Box\Diamond p_1\lor\Diamond\Box\Diamond p_2$;\\

\textbf{Stage 1}:\\

The distribution rules, the splitting rules and the monotone and antitone variable elimination rules are not applicable here, so\\

\begin{tabular}{r l}

& $\mathsf{Preprocess}(\downarrow x.(\Diamond\Box p_1\land\Box(@_\nomi x\land\Box p_2))\leq \Diamond\Box\Diamond p_1\lor\Diamond\Box\Diamond p_2)$\\
$=$ & $\{\downarrow x.(\Diamond\Box p_1\land\Box(@_\nomi x\land\Box p_2))\leq \Diamond\Box\Diamond p_1\lor\Diamond\Box\Diamond p_2\}$;\\

\end{tabular}
$\ $\\

then we apply the first approximation rule and get\\

$\{\nomi_{0}\leq\downarrow x.(\Diamond\Box p_1\land\Box(@_\nomi x\land\Box p_2)), \Diamond\Box\Diamond p_1\lor\Diamond\Box\Diamond p_2\leq \neg\nomi_{1}\}$;\\

\textbf{Stage 2}:\\

\textbf{Substage 1}:\\

by applying the splitting rule for $\lor$ we get\\

$\{\nomi_{0}\leq\downarrow x.(\Diamond\Box p_1\land\Box(@_\nomi x\land\Box p_2)), \Diamond\Box\Diamond p_1\leq\neg\nomi_{1}, \Diamond\Box\Diamond p_2\leq \neg\nomi_{1}\}$;\\

by applying the approximation rule for $\downarrow x$ we get\\

$\{\nomi_{0}\leq\Diamond\Box p_1\land\Box(@_\nomi \nomi_{0}\land\Box p_2), \Diamond\Box\Diamond p_1\leq\neg\nomi_{1}, \Diamond\Box\Diamond p_2\leq \neg\nomi_{1}\}$;\\

by applying the splitting rule for $\land$ we get\\

$\{\nomi_{0}\leq\Diamond\Box p_1,\nomi_{0}\leq\Box(@_\nomi \nomi_{0}\land\Box p_2), \Diamond\Box\Diamond p_1\leq\neg\nomi_{1}, \Diamond\Box\Diamond p_2\leq \neg\nomi_{1}\}$;\\

by applying the approximation rule for $\Diamond$ we get\\

$\{\nomi_{0}\leq\Diamond\nomi_{2},\nomi_{2}\leq\Box p_1,\nomi_{0}\leq\Box(@_\nomi \nomi_{0}\land\Box p_2), \Diamond\Box\Diamond p_1\leq\neg\nomi_{1}, \Diamond\Box\Diamond p_2\leq \neg\nomi_{1}\}$;\\

\textbf{Substage 2}:\\

by applying the residuation rule for $\Box$ twice we get\\

$\{\nomi_{0}\leq\Diamond\nomi_{2},\Diamondblack\nomi_{2}\leq p_1,\Diamondblack\nomi_{0}\leq @_\nomi \nomi_{0}\land\Box p_2, \Diamond\Box\Diamond p_1\leq\neg\nomi_{1}, \Diamond\Box\Diamond p_2\leq \neg\nomi_{1}\}$;\\

by applying the splitting rule for $\land$ we get\\

$\{\nomi_{0}\leq\Diamond\nomi_{2},\Diamondblack\nomi_{2}\leq p_1,\Diamondblack\nomi_{0}\leq @_\nomi \nomi_{0}, \Diamondblack\nomi_{0}\leq\Box p_2, \Diamond\Box\Diamond p_1\leq\neg\nomi_{1}, \Diamond\Box\Diamond p_2\leq \neg\nomi_{1}\}$;\\

by applying the residuation rule for $\Box$ we get\\

$\{\nomi_{0}\leq\Diamond\nomi_{2},\Diamondblack\nomi_{2}\leq p_1,\Diamondblack\nomi_{0}\leq @_\nomi \nomi_{0}, \Diamondblack\Diamondblack\nomi_{0}\leq p_2, \Diamond\Box\Diamond p_1\leq\neg\nomi_{1}, \Diamond\Box\Diamond p_2\leq \neg\nomi_{1}\}$;\\

by applying the residuation rule for $\Diamond$ twice we get\\

$\{\nomi_{0}\leq\Diamond\nomi_{2},\Diamondblack\nomi_{2}\leq p_1,\Diamondblack\nomi_{0}\leq @_\nomi \nomi_{0}, \Diamondblack\Diamondblack\nomi_{0}\leq p_2,\Box\Diamond p_1\leq\blacksquare\neg\nomi_{1}, \Box\Diamond p_2\leq\blacksquare \neg\nomi_{1}\}$;\\

it is easy to see that the second splitting rule is not applicable here;\\

\textbf{Substage 3}:\\

it is easy to see that the packing rule is not applicable here;\\

\textbf{Substage 4}:\\

by applying the right-handed Ackermann rule for $p_1$, we get\\

$\{\nomi_{0}\leq\Diamond\nomi_{2},\Diamondblack\nomi_{0}\leq @_\nomi \nomi_{0}, \Diamondblack\Diamondblack\nomi_{0}\leq p_2,\Box\Diamond \Diamondblack\nomi_{2}\leq\blacksquare\neg\nomi_{1}, \Box\Diamond p_2\leq\blacksquare \neg\nomi_{1}\}$;\\

by applying the right-handed Ackermann rule for $p_2$, we get\\

$\{\nomi_{0}\leq\Diamond\nomi_{2},\Diamondblack\nomi_{0}\leq @_\nomi \nomi_{0},\Box\Diamond \Diamondblack\nomi_{2}\leq\blacksquare\neg\nomi_{1}, \Box\Diamond\Diamondblack\Diamondblack\nomi_{0}\leq\blacksquare \neg\nomi_{1}\}$;\\

\textbf{Stage 3}:\\

Since all the propositional variables are eliminated, we output the quasi-inequality 

$(\nomi_{0}\leq\Diamond\nomi_{2})\ \&\ (\Diamondblack\nomi_{0}\leq @_\nomi \nomi_{0})\ \&\ (\Box\Diamond \Diamondblack\nomi_{2}\leq\blacksquare\neg\nomi_{1})\ \&\ (\Box\Diamond\Diamondblack\Diamondblack\nomi_{0}\leq\blacksquare\neg\nomi_{1})\ \Rightarrow (\nomi_{0}\leq\neg\nomi_{1})$;\\

then we get its standard translation, and use $\forall i_0\forall i_1\forall i_2\forall i$ to universally quantify it.

\end{example}

\begin{example}\label{Example:ALBA:execution:2}
We consider another $(1,1)$-Sahlqvist inequality $\Box\downarrow x.(\Box p_1\land\Box(x\land @_{\nomi}p_2))\leq\Diamond\Box p_1\lor\Diamond\Box p_2$, to show how to use the universal quantifiers and the residuation rules for $\downarrow$ and $@$:\\

\textbf{Stage 1}:\\

The distribution rules, the splitting rules and the monotone and antitone variable elimination rules are not applicable here, so\\

\begin{tabular}{r l}
& $\mathsf{Preprocess}(\Box\downarrow x.(\Box p_1\land\Box(x\land @_{\nomi}p_2))\leq\Diamond\Box p_1\lor\Diamond\Box p_2)$\\

$=$ & $\{\Box\downarrow x.(\Box p_1\land\Box(x\land @_{\nomi}p_2))\leq\Diamond\Box p_1\lor\Diamond\Box p_2\}$;\\
\end{tabular}

$\ $\\
then we apply the first approximation rule and get\\

$\{\nomi_{0}\leq\Box\downarrow x.(\Box p_1\land\Box(x\land @_{\nomi}p_2)),\Diamond\Box p_1\lor\Diamond\Box p_2\leq\neg\nomi_{1}\}$;\\

\textbf{Stage 2}:\\

\textbf{Substage 1}:\\

by applying the splitting rule for $\lor$ we get\\

$\{\nomi_{0}\leq\Box\downarrow x.(\Box p_1\land\Box(x\land @_{\nomi}p_2)),\Diamond\Box p_1\leq\neg\nomi_{1}, \Diamond\Box p_2\leq\neg\nomi_{1}\}$;\\

\textbf{Substage 2}:\\

by applying the residuation rule for $\Box$ and the residuation rule for $\Diamond$ twice we get\\

$\{\Diamondblack\nomi_{0}\leq\downarrow x.(\Box p_1\land\Box(x\land @_{\nomi}p_2)),\Box p_1\leq\blacksquare\neg\nomi_{1},\Box p_2\leq\blacksquare\neg\nomi_{1}\}$;\\

by applying the residuation rule for $\downarrow x$ we get\\

$\{\forall y(\mathsf{A}(y\to\Diamondblack\nomi_{0})\land y\leq\Box p_1\land\Box(y\land @_{\nomi}p_2)),\Box p_1\leq\blacksquare\neg\nomi_{1},\Box p_2\leq\blacksquare\neg\nomi_{1}\}$;\\

by applying the splitting rule for $\land$ we get (notice that here we apply the splitting rule in the scope of the universal quantifier $\forall y$)\\

$\{\forall y(\mathsf{A}(y\to\Diamondblack\nomi_{0})\land y\leq\Box p_1\ \&\ \mathsf{A}(y\to\Diamondblack\nomi_{0})\land y\leq\Box(y\land @_{\nomi}p_2)),\Box p_1\leq\blacksquare\neg\nomi_{1},\Box p_2\leq\blacksquare\neg\nomi_{1}\}$;\\

by applying the second splitting rule we get\\

$\{\forall y(\mathsf{A}(y\to\Diamondblack\nomi_{0})\land y\leq\Box p_1),\forall y(\mathsf{A}(y\to\Diamondblack\nomi_{0})\land y\leq\Box(y\land @_{\nomi}p_2)),\Box p_1\leq\blacksquare\neg\nomi_{1},\Box p_2\leq\blacksquare\neg\nomi_{1}\}$;\\

by applying the residuation rule for $\Box$ twice we get (notice that here we apply the residuation rule in the scope of the universal quantifier $\forall y$)\\

$\{\forall y(\Diamondblack(\mathsf{A}(y\to\Diamondblack\nomi_{0})\land y)\leq p_1),\forall y(\Diamondblack(\mathsf{A}(y\to\Diamondblack\nomi_{0})\land y)\leq y\land @_{\nomi}p_2),\Box p_1\leq\blacksquare\neg\nomi_{1},\Box p_2\leq\blacksquare\neg\nomi_{1}\}$;\\

by applying the splitting rule for $\land$ we get (notice that here we apply the splitting rule in the scope of the universal quantifier $\forall y$)\\

$\{\forall y(\Diamondblack(\mathsf{A}(y\to\Diamondblack\nomi_{0})\land y)\leq p_1),\forall y(\Diamondblack(\mathsf{A}(y\to\Diamondblack\nomi_{0})\land y)\leq y\ \&\ \Diamondblack(\mathsf{A}(y\to\Diamondblack\nomi_{0})\land y)\leq @_{\nomi}p_2),\Box p_1\leq\blacksquare\neg\nomi_{1},\Box p_2\leq\blacksquare\neg\nomi_{1}\}$;\\

by applying the second splitting rule we get\\

$\{\forall y(\Diamondblack(\mathsf{A}(y\to\Diamondblack\nomi_{0})\land y)\leq p_1),\forall y(\Diamondblack(\mathsf{A}(y\to\Diamondblack\nomi_{0})\land y)\leq y),\forall y(\Diamondblack(\mathsf{A}(y\to\Diamondblack\nomi_{0})\land y)\leq @_{\nomi}p_2),\Box p_1\leq\blacksquare\neg\nomi_{1},\Box p_2\leq\blacksquare\neg\nomi_{1}\}$;\\

by applying the residuation rule for $@$ we get\\

$\{\forall y(\Diamondblack(\mathsf{A}(y\to\Diamondblack\nomi_{0})\land y)\leq p_1),\forall y(\Diamondblack(\mathsf{A}(y\to\Diamondblack\nomi_{0})\land y)\leq y),\forall y(\mathsf{E}\Diamondblack(\mathsf{A}(y\to\Diamondblack\nomi_{0})\land y)\land\nomi\leq p_2),\Box p_1\leq\blacksquare\neg\nomi_{1},\Box p_2\leq\blacksquare\neg\nomi_{1}\}$;\\

\textbf{Substage 3}:\\

by applying the packing rule twice we get\\

$\{\exists y(\Diamondblack(\mathsf{A}(y\to\Diamondblack\nomi_{0})\land y))\leq p_1,\forall y(\Diamondblack(\mathsf{A}(y\to\Diamondblack\nomi_{0})\land y)\leq y),\exists y(\mathsf{E}\Diamondblack(\mathsf{A}(y\to\Diamondblack\nomi_{0})\land y)\land\nomi)\leq p_2,\Box p_1\leq\blacksquare\neg\nomi_{1},\Box p_2\leq\blacksquare\neg\nomi_{1}\}$;\\

\textbf{Substage 4}:\\

by applying the right-handed Ackermann rule for $p_1$, we get\\

$\{\forall y(\Diamondblack(\mathsf{A}(y\to\Diamondblack\nomi_{0})\land y)\leq y),\exists y(\mathsf{E}\Diamondblack(\mathsf{A}(y\to\Diamondblack\nomi_{0})\land y)\land\nomi)\leq p_2,\Box \exists y(\Diamondblack(\mathsf{A}(y\to\Diamondblack\nomi_{0})\land y))\leq\blacksquare\neg\nomi_{1},\Box p_2\leq\blacksquare\neg\nomi_{1}\}$;\\

by applying the right-handed Ackermann rule for $p_2$, we get\\

$\{\forall y(\Diamondblack(\mathsf{A}(y\to\Diamondblack\nomi_{0})\land y)\leq y),\Box\exists y(\Diamondblack(\mathsf{A}(y\to\Diamondblack\nomi_{0})\land y))\leq\blacksquare\neg\nomi_{1},\Box\exists y(\mathsf{E}\Diamondblack(\mathsf{A}(y\to\Diamondblack\nomi_{0})\land y)\land\nomi)\leq\blacksquare\neg\nomi_{1}\}$;\\

\textbf{Stage 3}:\\

Since all the propositional variables are eliminated, we output the quasi-inequality 

$(\forall y(\Diamondblack(\mathsf{A}(y\to\Diamondblack\nomi_{0})\land y)\leq y))\ \&\ (\Box\exists y(\Diamondblack(\mathsf{A}(y\to\Diamondblack\nomi_{0})\land y))\leq\blacksquare\neg\nomi_{1})\ \&\ (\Box\exists y(\mathsf{E}\Diamondblack(\mathsf{A}(y\to\Diamondblack\nomi_{0})\land y)\land\nomi)\leq\blacksquare\neg\nomi_{1})\ \Rightarrow (\nomi_{0}\leq\neg\nomi_{1})$;\\

then we get its standard translation, and use $\forall i_0\forall i_1\forall i$ to universally quantify it.
\end{example}

\section{Soundness of $\mathsf{ALBA}^{\downarrow}$}\label{Sec:Soundness}
In the present section, we will prove the soundness of the algorithm $\mathsf{ALBA}^{\downarrow}$ with respect to Kripke frames. The basic proof structure is similar to \cite{Zh21}.

\begin{theorem}[Soundness]\label{Thm:Soundness}
If $\mathsf{ALBA}^{\downarrow}$ runs successfully on $\phi\leq\psi$ and outputs $\mathsf{FO}(\phi\leq\psi)$, then for any Kripke frame $\mathbb{F}=(W,R)$, $$\mathbb{F}\Vdash\phi \leq \psi\mbox{ iff }\mathbb{F}\models\mathsf{FO}(\phi \leq \psi).$$
\end{theorem}

\begin{proof}
The proof goes similarly to \cite[Theorem 8.1]{CoPa12}. Let $\phi_i\leq\psi_i$, $1\leq i\leq n$ denote the inequalities produced by preprocessing $\phi\leq\psi$ after Stage 1, and $\{\nomi_{0}\leq\phi_{i}, \psi_{i}\leq\neg\nomi_{1}\}$ denote the inequalities after the first-approximation rule, $\mathsf{Reduce}(\phi_i\leq\psi_i)$ denote the set of pure (universally quantified) inequalities after Stage 2, and $\mathsf{FO}(\phi \leq \psi)$ denote the standard translation of the quasi-(universally quantified) inequalities into first-order formulas, then we have the following chain of equivalences:

It suffices to show the equivalence from (\ref{aCrct:Eqn0}) to (\ref{aCrct:Eqn4}) given below:

\begin{eqnarray}
&&\mathbb{F}\Vdash\phi\leq\psi\label{aCrct:Eqn0}\\
&&\mathbb{F}\Vdash\phi_i\leq\psi_i,\mbox{ for all }1\leq i\leq n\label{aCrct:Eqn1}\\
&&\mathbb{F}\Vdash(\nomi_{0}\leq\phi_{i}\ \&\ \psi_{i}\leq\neg\nomi_{1})\Rightarrow\nomi_{0}\leq\neg\nomi_{1} \mbox{ for all }1\leq i\leq n\label{aCrct:Eqn2}\\
&&\mathbb{F}\Vdash\mathsf{Reduce}(\phi_{i}\leq\psi_{i})\Rightarrow\nomi_{0}\leq\neg\nomi_{1} \mbox{ for all }1\leq i\leq n\label{aCrct:Eqn3}\\
&&\mathbb{F}\Vdash\mathsf{FO}(\phi\leq\psi)\label{aCrct:Eqn4}
\end{eqnarray}

\begin{itemize}
\item The equivalence between (\ref{aCrct:Eqn0}) and (\ref{aCrct:Eqn1}) follows from Proposition \ref{prop:Soundness:stage:1};
\item the equivalence between (\ref{aCrct:Eqn1}) and (\ref{aCrct:Eqn2}) follows from Proposition \ref{prop:Soundness:first:approximation};
\item the equivalence between (\ref{aCrct:Eqn2}) and (\ref{aCrct:Eqn3}) follows from Propositions \ref{Prop:Substage:1}, \ref{Prop:Substage:2}, \ref{Prop:Substage:3}, \ref{Prop:Substage:4};
\item the equivalence between (\ref{aCrct:Eqn3}) and (\ref{aCrct:Eqn4}) follows from Proposition \ref{Prop:ST:ineq:quasi:mega}.
\end{itemize}
\end{proof}

In the remainder of this section, we prove the soundness of the rules in Stage 1, 2 and 3.

\begin{proposition}[Soundness of the rules in Stage 1]\label{prop:Soundness:stage:1}
For the distribution rules, the splitting rules and the monotone and antitone variable-elimination rules, they are sound in both directions in $\mathbb{F}$, i.e.\ the inequality before the rule is valid in $\mathbb{F}$ iff the inequality(-ies) after the rule is(are) valid in $\mathbb{F}$.
\end{proposition}

\begin{proof}
For the soundness of the distribution rules, it follows from the fact that the following equivalences are valid in $\mathbb{F}$:

\begin{itemize}
\item $\Diamond(\alpha\lor\beta)\leftrightarrow\Diamond\alpha\lor\Diamond\beta$;
\item $\neg(\alpha\lor\beta)\leftrightarrow\neg\alpha\land\neg\beta$;
\item $(\alpha\lor\beta)\land\gamma\leftrightarrow(\alpha\land\gamma)\lor(\beta\land\gamma)$;
\item $\alpha\land(\beta\lor\gamma)\leftrightarrow(\alpha\land\beta)\lor(\alpha\land\gamma)$;
\item $\downarrow x.(\alpha\lor\beta)\leftrightarrow (\downarrow x.\alpha\lor\downarrow x.\beta)$;
\item $@_{\nomi}(\alpha\lor\beta)\leftrightarrow (@_{\nomi}\alpha\lor @_{\nomi}\beta)$;
\item $@_{x}(\alpha\lor\beta)\leftrightarrow (@_{x}\alpha\lor @_{x}\beta)$;
\item $((\alpha\lor\beta)\to\gamma)\leftrightarrow((\alpha\to\gamma)\land(\beta\to\gamma))$;
\item $\Box(\alpha\land\beta)\leftrightarrow\Box\alpha\land\Box\beta$;
\item $\neg(\alpha\land\beta)\leftrightarrow\neg\alpha\lor\neg\beta$;
\item $(\alpha\land\beta)\lor\gamma\leftrightarrow(\alpha\lor\gamma)\land(\beta\lor\gamma)$;
\item $\alpha\lor(\beta\land\gamma)\leftrightarrow(\alpha\lor\beta)\land(\alpha\lor\gamma)$;
\item $\downarrow x.(\alpha\land\beta)\leftrightarrow (\downarrow x.\alpha\land \downarrow x.\beta)$;
\item $@_{\nomi}(\alpha\land\beta)\leftrightarrow (@_{\nomi}\alpha\land @_{\nomi}\beta)$;
\item $@_{x}(\alpha\land\beta)\leftrightarrow (@_{x}\alpha\land @_{x}\beta)$;
\item $(\alpha\to\beta\land\gamma)\leftrightarrow(\alpha\to\beta)\land(\alpha\to\gamma)$.
\end{itemize}

For the soundness of the splitting rules, it follows from the following fact:

$$\mathbb{F}\Vdash\alpha\leq\beta\land\gamma\mbox{ iff }(\mathbb{F}\Vdash\alpha\leq\beta\mbox{ and }\mathbb{F}\Vdash\alpha\leq\gamma);$$
$$\mathbb{F}\Vdash\alpha\lor\beta\leq\gamma\mbox{ iff }(\mathbb{F}\Vdash\alpha\leq\gamma\mbox{ and }\mathbb{F}\Vdash\beta\leq\gamma).$$

For the soundness of the monotone and antitone variable elimination rules, we show the soundness for the first rule. Suppose $\alpha(p)$ is negative in $p$ and $\beta$ is positive in $p$. 

If $\mathbb{F}\Vdash\alpha(p)\leq\beta(p)$, then for all valuations $V$ and all assignments $g$, $(\mathbb{F}, V),g\Vdash\alpha(p)\leq\beta(p)$, thus for the valuation $V^{p}_{\emptyset}$ such that $V^{p}_{\emptyset}$ is the same as $V$ except that $V^{p}_{\emptyset}(p)=\emptyset$, $(\mathbb{F}, V^{p}_{\emptyset}),g\Vdash\alpha(p)\leq\beta(p)$, therefore $(\mathbb{F}, V^{p}_{\emptyset}),g\Vdash\alpha(\bot)\leq\beta(\bot)$, thus $(\mathbb{F}, V),g\Vdash\alpha(\bot)\leq\beta(\bot)$, so $\mathbb{F}\Vdash\alpha(\bot)\leq\beta(\bot)$.

For the other direction, suppose $\mathbb{F}\vDash\alpha(\bot)\leq\beta(\bot)$, then by the fact that $\alpha(p)$ is negative in $p$ and $\beta$ is positive in $p$, we have that $\mathbb{F}\vDash\alpha(p)\leq\alpha(\bot)$ and $\mathbb{F}\vDash\beta(\bot)\leq\beta(p)$, therefore $\mathbb{F}\vDash\alpha(p)\leq\beta(p)$.

The soundness of the other rule is similar.
\end{proof}

\begin{proposition}\label{prop:Soundness:first:approximation}
(\ref{aCrct:Eqn1}) and (\ref{aCrct:Eqn2}) are equivalent, i.e.\ the first-approximation rule is sound in $\mathbb{F}$.
\end{proposition}

\begin{proof}
(\ref{aCrct:Eqn1}) $\Rightarrow$ (\ref{aCrct:Eqn2}): Suppose $\mathbb{F}\Vdash\phi_i\leq\psi_i$. Then for any valuation $V$ and any assignment $g$ on $\mathbb{F}$, if $(\mathbb{F},V),g\Vdash\nomi_{0}\leq\phi_{i}$ and $(\mathbb{F},V),g\Vdash\psi_{i}\leq\neg\nomi_{1}$, then $(\mathbb{F},V),g,V(\nomi_{0})\Vdash\phi_{i}$ and $(\mathbb{F},V),g,V(\nomi_{1})\nVdash\psi_{i}$, so by $\mathbb{F}\Vdash\phi_{i}\leq\psi_{i}$ we have $(\mathbb{F},V),g,V(\nomi_{0})\Vdash\psi_{i}$, so $V(\nomi_{0})\neq V(\nomi_{1})$, so $(\mathbb{F},V),g\Vdash\nomi_{0}\leq\neg\nomi_{1}$.

(\ref{aCrct:Eqn2}) $\Rightarrow$ (\ref{aCrct:Eqn1}): Suppose $\mathbb{F}\Vdash(\nomi_{0}\leq\phi_{i}\ \&\ \psi_{i}\leq\neg\nomi_{1})\Rightarrow\nomi_{0}\leq\neg\nomi_{1}$. Then if $\mathbb{F}\nVdash\phi_i\leq\psi_i$, there is a valuation $V$ and an assignment $g$ on $\mathbb{F}$ and a $w\in W$ such that $(\mathbb{F}, V),g, w\Vdash\phi_i$ and $(\mathbb{F}, V),g, w\nVdash\psi_i$. Then by taking $V^{\nomi_{0}, \nomi_{1}}_{w, w}$ to be the valuation which is the same as $V$ except that $V^{\nomi_{0}, \nomi_{1}}_{w, w}(\nomi_{0})=V^{\nomi_{0}, \nomi_{1}}_{w, w}(\nomi_{1})=\{w\}$, then since $\nomi_{0}, \nomi_{1}$ do not occur in $\phi_{i}$ and $\psi_{i}$, we have that $(\mathbb{F}, V^{\nomi_{0}, \nomi_{1}}_{w, w}),g, w\Vdash\phi_i$ and $(\mathbb{F}, V^{\nomi_{0}, \nomi_{1}}_{w, w}),g, w\nVdash\psi_i$, therefore $(\mathbb{F}, V^{\nomi_{0}, \nomi_{1}}_{w, w}),g\Vdash\nomi_{0}\leq\phi_i$ and $(\mathbb{F}, V^{\nomi_{0}, \nomi_{1}}_{w, w}),g\Vdash\psi_i\leq\neg\nomi_{1}$, by $\mathbb{F}\Vdash(\nomi_{0}\leq\phi_{i}\ \&\ \psi_{i}\leq\neg\nomi_{1})\Rightarrow\nomi_{0}\leq\neg\nomi_{1}$, we have that $(\mathbb{F}, V^{\nomi_{0}, \nomi_{1}}_{w, w}),g\Vdash\nomi_{0}\leq\neg\nomi_{1}$, so $(\mathbb{F}, V^{\nomi_{0}, \nomi_{1}}_{w, w}),g,w\Vdash\nomi_{0}$ implies that $(\mathbb{F}, V^{\nomi_{0}, \nomi_{1}}_{w, w}),g,w\Vdash\neg\nomi_{1}$, a contradiction. So $\mathbb{F}\Vdash\phi_i\leq\psi_i$.
\end{proof}

The next step is to show the soundness of each rule of Stage 2. For each rule, before the application of this rule we have a set of mega-inequalities $S$ (which we call the \emph{system}), after applying the rule we get a set of mega-inequalities $S'$, the soundness of Stage 2 is then the equivalence of the following two conditions:
\begin{itemize}
\item $\mathbb{F}\Vdash\bigamp S\ \Rightarrow \nomi_0\leq\neg\nomi_1$;

\item $\mathbb{F}\Vdash\bigamp S'\ \Rightarrow \nomi_0\leq\neg\nomi_1$;
\end{itemize}

where $\bigamp S$ denote the meta-conjunction of mega-inequalities of $S$. For substage 1 and substage 4, it suffices to show the following property:

\begin{itemize}\label{condition:1:4:equivalence}
\item For any frame $\mathbb{F}$, any valuation $V$ and any assignment $g$ on it, if $(\mathbb{F},V),g\Vdash S$, then there are valuation $V'$ and assignment $g'$ such that $V'(\nomi_0)=V(\nomi_0)$, $V'(\nomi_1)=V(\nomi_1)$ and $(\mathbb{F},V'),g'\Vdash S'$;
\item For any frame $\mathbb{F}$, any valuation $V'$ and any assignment $g'$ on it, if $(\mathbb{F},V'),g'\Vdash S'$, then there are valuation $V$ and assignment $g$ such that $V(\nomi_0)=V'(\nomi_0)$, $V(\nomi_1)=V'(\nomi_1)$ and $(\mathbb{F},V),g\Vdash S$.
\end{itemize}

For substage 2 and 3, due to the involvement of universal quantifiers, we use the following stronger property. Assume the premise of a rule is $T$, and the conclusions of the rule is $T'$, it suffices to show the following property:

$$(\mathbb{F},V),g\Vdash T\mbox{ iff }(\mathbb{F},V),g\Vdash T'.$$

\begin{proposition}\label{Prop:Substage:1}
The splitting rules, the approximation rules for $\Diamond,\Box,@,\downarrow,\to$, the residuation rules for $\neg$ in Substage 1 are sound in $\mathbb{F}$.
\end{proposition}

\begin{proof}
For the splitting rules, the approximation rules for $\Diamond,\Box,\to$, the residuation rules for $\neg$ in Substage 1, their soundness proofs are given by Lemma \ref{Lemma:Splitting:substage12}, \ref{Lemma:approximation:white:substage1}, \ref{Lemma:Approximation:to:substage1} and \ref{Lemma:residuation:neg:substage12} below. The soundness of the approximation rules for $@$ and $\downarrow$ are proved in Lemma \ref{Lemma:Substage:1:at} and \ref{Lemma:Substage:1:downarrow}.
\end{proof}

\begin{lemma}\label{Lemma:Splitting:substage12}
The splitting rules in Substage 1 and Substage 2 are sound in $\mathbb{F}$.
\end{lemma}

\begin{proof}
For the soundness of the splitting rules, it follows from the fact that for any Kripke frame $\mathbb{F}=(W,R)$, any valuation $V$ and any assignment $g$ on $\mathbb{F}$, 
\begin{itemize}
\item $(\mathbb{F},V),g\Vdash\alpha\leq\beta\land\gamma\mbox{ iff }((\mathbb{F},V),g\Vdash\alpha\leq\beta\mbox{ and }(\mathbb{F},V),g\Vdash\alpha\leq\gamma),$
\item $(\mathbb{F},V),g\Vdash\alpha\lor\beta\leq\gamma\mbox{ iff }((\mathbb{F},V),g\Vdash\alpha\leq\gamma\mbox{ and }(\mathbb{F},V),g\Vdash\beta\leq\gamma).$
\end{itemize}
\end{proof}

\begin{lemma}\label{Lemma:approximation:white:substage1}
The approximation rules for $\Diamond,\Box$ in Substage 1 are sound in $\mathbb{F}$.
\end{lemma}

\begin{proof}
We prove for $\Diamond$ and nominals, the cases for $\Box$ and state variables are similar.

For the soundness of the approximation rule for $\Diamond$, it suffices to show that for any Kripke frame $\mathbb{F}=(W,R)$, any valuation $V$ and any assignment $g$ on $\mathbb{F}$, 
\begin{enumerate}
\item if $(\mathbb{F},V),g\Vdash\nomi\leq\Diamond\alpha$, then there is a valuation $V^{\nomj}$ such that $V^{\nomj}$ is the same as $V$ except $V^{\nomj}(\nomj)$, and $(\mathbb{F},V^{\nomj}),g\Vdash\nomi\leq\Diamond\nomj$ and $(\mathbb{F},V^{\nomj}),g\Vdash\nomj\leq\alpha$;
\item if $(\mathbb{F},V),g\Vdash\nomi\leq\Diamond\nomj$ and $(\mathbb{F},V),g\Vdash\nomj\leq\alpha$, then $(\mathbb{F},V),g\Vdash\nomi\leq\Diamond\alpha$.
\end{enumerate}
For item 1, if $(\mathbb{F},V),g\Vdash\nomi\leq\Diamond\alpha$, then $(\mathbb{F},V),g, V(\nomi)\Vdash\Diamond\alpha$, therefore there exists a $w\in W$ such that $(V(\nomi),w)\in R$ and $(\mathbb{F},V),g, w\Vdash\alpha$. Now take $V^{\nomj}$ such that $V^{\nomj}$ is the same as $V$ except that $V^{\nomj}(\nomj)=\{w\}$, then $(V^{\nomj}(\nomi), V^{\nomj}(\nomj))\in R$, so $(\mathbb{F},V^{\nomj}),g\Vdash\nomi\leq\Diamond\nomj$ and $(\mathbb{F},V^{\nomj}),g\Vdash\nomj\leq\alpha$.

For item 2, suppose $(\mathbb{F},V),g\Vdash\nomi\leq\Diamond\nomj$ and $(\mathbb{F},V),g\Vdash\nomj\leq\alpha$. Then $(V(\nomi), V(\nomj))\in R$ and $(\mathbb{F},V),g,V(\nomj)\Vdash\alpha$, so $(\mathbb{F},V),g,V(\nomi)\Vdash\Diamond\alpha$, therefore $(\mathbb{F},V),g\Vdash\nomi\leq\Diamond\alpha$.
\end{proof}

\begin{lemma}\label{Lemma:Approximation:to:substage1}
The approximation rule for $\to$ in Substage 1 is sound in $\mathbb{F}$.
\end{lemma}

\begin{proof}
We only prove it for the nominal version, the state variable version is similar. 

For the soundness of the approximation rule for $\to$, it suffices to show that for any Kripke frame $\mathbb{F}=(W,R)$, any valuation $V$ and any assignment $g$ on $\mathbb{F}$, 
\begin{enumerate}
\item if $(\mathbb{F},V),g\Vdash\alpha\to\beta\leq\neg\nomi$, then there is a valuation $V^{\nomj,\nomk}$ such that $V^{\nomj,\nomk}$ is the same as $V$ except $V^{\nomj,\nomk}(\nomj)$ and $V^{\nomj,\nomk}(\nomk)$, and $(\mathbb{F},V^{\nomj,\nomk}),g\Vdash\nomj\leq\alpha$, $(\mathbb{F},V^{\nomj,\nomk}),g\Vdash\beta\leq\neg \nomk$ and $(\mathbb{F},V^{\nomj,\nomk}),g\Vdash\nomj\to\neg\nomk\leq\neg\nomi$;
\item if $(\mathbb{F},V),g\Vdash\nomj\leq\alpha$, $(\mathbb{F},V),g\Vdash\beta\leq\neg \nomk$ and $(\mathbb{F},V),g\Vdash\nomj\to\neg\nomk\leq\neg\nomi$, then $(\mathbb{F},V),g\Vdash\alpha\to\beta\leq\neg\nomi$.
\end{enumerate}

For item 1, if $(\mathbb{F},V),g\Vdash\alpha\to\beta\leq\neg\nomi$, then $(\mathbb{F},V),g, V(\nomi)\Vdash\alpha$ and $(\mathbb{F},V),g, V(\nomi)\Vdash\neg\beta$. Now take $V^{\nomj,\nomk}$ such that $V^{\nomj,\nomk}$ is the same as $V$ except that $V^{\nomj,\nomk}(\nomj)=V^{\nomj,\nomk}(\nomk)=V(\nomi)$, we have that $(\mathbb{F},V^{\nomj,\nomk}),g, V^{\nomj,\nomk}(\nomj)\Vdash\alpha$ and $(\mathbb{F},V^{\nomj,\nomk}),g, V^{\nomj,\nomk}(\nomk)\Vdash\neg\beta$, so $(\mathbb{F},V^{\nomj,\nomk}),g\Vdash\nomj\leq\alpha$, $(\mathbb{F},V^{\nomj,\nomk}),g\Vdash\beta\leq\neg \nomk$. Since $V^{\nomj,\nomk}(\nomj)=V^{\nomj,\nomk}(\nomk)=V^{\nomj,\nomk}(\nomi)=V(\nomi)$, it is easy to see that $V^{\nomj,\nomk}(\nomj\to\neg\nomk)=V^{\nomj,\nomk}(\neg\nomi)$, so $(\mathbb{F},V^{\nomj,\nomk}),g\Vdash\nomj\to\neg\nomk\leq\neg\nomi$.

For item 2, if $(\mathbb{F},V),g\Vdash\nomj\leq\alpha$, $(\mathbb{F},V),g\Vdash\beta\leq\neg\nomk$ and $(\mathbb{F},V),g\Vdash\nomj\to\neg\nomk\leq\neg\nomi$, then $V(\nomj\to\neg\nomk)\subseteq V(\neg\nomi)$, so $V(\nomi)\subseteq V(\nomj\land\nomk)$, since $\nomi,\nomj,\nomk$ are nominals, there interpretations are singletons, so $V(\nomi)=V(\nomj)=V(\nomk)$. Now from $(\mathbb{F},V),g\Vdash\nomj\leq\alpha$ we have that $(\mathbb{F},V),g, V(\nomj)\Vdash\alpha$, and from $(\mathbb{F},V),g\Vdash\beta\leq\neg\nomk$ we have that $(\mathbb{F},V),g, V(\nomk)\Vdash\neg\beta$, so $(\mathbb{F},V),g, V(\nomi)\Vdash\alpha$ and $(\mathbb{F},V),g, V(\nomi)\Vdash\neg\beta$, so $(\mathbb{F},V),g\Vdash\alpha\to\beta\leq\neg\nomi$.
\end{proof}

\begin{lemma}\label{Lemma:residuation:neg:substage12}
The residuation rules for $\neg$ in Substage 1 and 2 are sound in $\mathbb{F}$.
\end{lemma}

\begin{proof}
It is easy to see that the residuation rules for $\neg$ in Substage 1 are special cases of the residuation rules for $\neg$ in Substage 2 (modulo double negation elimination). Now we only prove it for the residuation rule in Substage 2 where negation symbols occur on the right-hand side of the inequalities, the other rule is similar. 

For the soundness of the residuation rule for $\neg$, it suffices to show that for any Kripke frame $\mathbb{F}=(W,R)$, any valuation $V$ and any assignment $g$ on $\mathbb{F}$, $(\mathbb{F},V),g\Vdash\alpha\leq\neg\beta$ iff $(\mathbb{F},V),g\Vdash\beta\leq\neg\alpha$. Indeed, it follows from the following equivalence:
\begin{center}
\begin{tabular}{c l}
& $(\mathbb{F},V),g\Vdash\alpha\leq\neg\beta$\\
iff & for all $w\in W$, if $(\mathbb{F},V),g,w\Vdash\alpha$, then $(\mathbb{F},V),g,w\nVdash\beta$\\
iff & for all $w\in W$, if $(\mathbb{F},V),g,w\Vdash\beta$, then $(\mathbb{F},V),g,w\nVdash\alpha$\\
iff & $(\mathbb{F},V),g\Vdash\beta\leq\neg\alpha$.\\
\end{tabular}
\end{center}
\end{proof}

\begin{lemma}\label{Lemma:Substage:1:at}
The approximation rules for $@$ in Substage 1 are sound in $\mathbb{F}$.
\end{lemma}

\begin{proof}
Since for $@_{\nomj}$ and $@_{x}$, for $\nomi$ and $y$, the proofs are essentially the same, we only prove it for $@_{\nomj}$ and $\nomi$. 

For the left approximation rule for $@_{\nomj}$, it suffices to show that for any Kripke model $\mathbb{M}=(W,R,V)$, any assignment $g$ on $\mathbb{M}$, 

\begin{enumerate}
\item if $\mathbb{M},g\Vdash \nomi\leq @_{\nomj}\alpha$, then $\mathbb{M},g\Vdash \nomj\leq\alpha$;
\item if $\mathbb{M},g\Vdash \nomj\leq\alpha$, then $\mathbb{M},g\Vdash \nomi\leq @_{\nomj}\alpha$.
\end{enumerate}

For item 1, if $\mathbb{M},g\Vdash \nomi\leq @_{\nomj}\alpha$, then $\mathbb{M},g, V(\nomi)\Vdash @_{\nomj}\alpha$, therefore $\mathbb{M},g, V(\nomj)\Vdash \alpha$, thus $\mathbb{M},g\Vdash \nomj\leq\alpha$.

For item 2, if $\mathbb{M},g\Vdash \nomj\leq\alpha$, then $\mathbb{M},g,V(\nomj)\Vdash\alpha$, so $\mathbb{M},g\Vdash @_{\nomj}\alpha$, therefore $\mathbb{M},g,V(\nomi)\Vdash @_{\nomj}\alpha$, thus 
$\mathbb{M},g\Vdash \nomi\leq @_{\nomj}\alpha$.

The right approximation rule for $@_{\nomj}$ is similar.
\end{proof}

\begin{lemma}\label{Lemma:Substage:1:downarrow}
The approximation rules for $\downarrow$ in Substage 1 are sound in $\mathbb{F}$.
\end{lemma}

\begin{proof}
We prove it for the left approximation rule, the right rule is similar.

For the left approximation rule for $\downarrow x$, it suffices to show that for any Kripke model $\mathbb{M}=(W,R,V)$, any assignment $g$ on $\mathbb{M}$, 

$$\mathbb{M},g\Vdash \nomi\leq \downarrow x.\alpha\mbox{ iff }\mathbb{M},g\Vdash \nomi\leq\alpha[\nomi/x].$$
Indeed, 

\begin{center}
\begin{tabular}{c l}
& $\mathbb{M},g\Vdash \nomi\leq \downarrow x.\alpha$\\
iff & $\mathbb{M},g,V(\nomi)\Vdash \downarrow x.\alpha$\\
iff & $\mathbb{M},g^{x}_{V(\nomi)},V(\nomi)\Vdash\alpha$\\
iff & $\mathbb{M},g^{x}_{V(\nomi)},V(\nomi)\Vdash\alpha[\nomi/x]$\\
iff & $\mathbb{M},g,V(\nomi)\Vdash\alpha[\nomi/x]$\\
iff & $\mathbb{M},g\Vdash\nomi\leq\alpha[\nomi/x]$.
\end{tabular}
\end{center}
\end{proof}

\begin{proposition}\label{Prop:Substage:2}
The splitting rules, the residuation rules for $\neg,\Diamond,\Box,@,\downarrow$, the second splitting rule in Substage 2 are sound in $\mathbb{F}$.
\end{proposition}

\begin{proof}
The soundness proofs of the splitting rules are given in Lemma \ref{Lemma:Splitting:substage12}. The soundness proofs of the residuation rules for $\neg$ are given in Lemma \ref{Lemma:residuation:neg:substage12}. The soundness proofs of the residuation rules for $\Diamond,\Box$ in Substage 2 are given in Lemma \ref{Lemma:Residuation:white:substage2}. The soundness of the residuation rules for $@$ and $\downarrow$ are proved in Lemma \ref{Lemma:at:substage2} and \ref{Lemma:downarrow:substage2}, and the soundness proof of the second splitting rule is in Lemma \ref{Lemma:2ndsplitting:substage2}. 
\end{proof}

\begin{lemma}\label{Lemma:Residuation:white:substage2}
The residuation rules for $\Diamond,\Box$ in Substage 2 are sound in $\mathbb{F}$. 
\end{lemma}

\begin{proof}
We prove it for $\Diamond$, and the rule for $\Box$ is similar.

To show the soundness of the residuation rule for $\Diamond$ in Substage 2, it suffices to show that for any Kripke frame $\mathbb{F}=(W,R)$, any valuation $V$ and any assignment $g$ on $\mathbb{F}$, $(\mathbb{F},V),g\Vdash\Diamond\alpha\leq\beta$ iff $(\mathbb{F},V),g\Vdash\alpha\leq\blacksquare\beta$.

$\Rightarrow$: if $(\mathbb{F},V),g\Vdash\Diamond\alpha\leq\beta$, then for all $w\in W$, if $(\mathbb{F},V),g,w\Vdash\Diamond\alpha$, then $(\mathbb{F},V),g,w\Vdash\beta$. Our aim is to show that for all $v\in W$, if $(\mathbb{F},V),g,v\Vdash\alpha$, then $(\mathbb{F},V),g,v\Vdash\blacksquare\beta$.

Consider any $v\in W$ such that $(\mathbb{F},V),g,v\Vdash\alpha$. Then for any $u\in W$ such that $(u,v)\in R$, $(\mathbb{F},V),g,u\Vdash\Diamond\alpha$. Since $(\mathbb{F},V),g\Vdash\Diamond\alpha\leq\beta$, we have that $(\mathbb{F},V),g,u\Vdash\beta$, so for any $u\in W$ such that $(v,u)\in R^{-1}$, $(\mathbb{F},V),g,u\Vdash\beta$, so $(\mathbb{F},V),g,v\Vdash\blacksquare\beta$.

$\Leftarrow$: if $(\mathbb{F},V),g\Vdash\alpha\leq\blacksquare\beta$, then for all $w\in W$, if $(\mathbb{F},V),g,w\Vdash\alpha$, then $(\mathbb{F},V),g,w\Vdash\blacksquare\beta$. Our aim is to show that for all $v\in W$, if $(\mathbb{F},V),g,v\Vdash\Diamond\alpha$, then $(\mathbb{F},V),g,v\Vdash\beta$.

Now assume that $(\mathbb{F},V),g,v\Vdash\Diamond\alpha$. Then there is a $u\in W$ such that $(v,u)\in R$ and $(\mathbb{F},V),g,u\Vdash\alpha$. By $(\mathbb{F},V),g\Vdash\alpha\leq\blacksquare\beta$, we have that $(\mathbb{F},V),g,u\Vdash\blacksquare\beta$. Therefore, for $v\in W$, we have $(u,v)\in R^{-1}$, thus $(\mathbb{F},V),g,v\Vdash\beta$.
\end{proof}

\begin{lemma}\label{Lemma:at:substage2}
The residuation rules for $@$ are sound in $\mathbb{F}$.
\end{lemma}

\begin{proof}
Since for $@_{\nomi}$ and $@_{x}$, the proofs are essentially the same, we only prove it for $@_{\nomi}$. 

For the left residuation rule for $@_{\nomi}$, it suffices to show that for any Kripke model $\mathbb{M}=(W,R,V)$, any assignment $g$ on $\mathbb{M}$, 

$$\mathbb{M},g\Vdash\alpha\leq @_{\nomj}\beta\mbox{ iff }\mathbb{M},g\Vdash\mathsf{E}\alpha\land\nomj\leq\beta.$$

Indeed, 

\begin{center}
\begin{tabular}{c l}
& $\mathbb{M},g\Vdash \alpha\leq @_{\nomj}\beta$\\
iff & $\mathbb{M},g\Vdash \alpha\leq \mathsf{A}(\nomj\to\beta)$\\
iff & $\mathbb{M},g\Vdash \mathsf{E}\alpha\leq\nomj\to\beta$\\
iff & $\mathbb{M},g\Vdash\mathsf{E}\alpha\land\nomj\leq\beta$.
\end{tabular}
\end{center}

For the right residuation rule for $@_{\nomi}$, it suffices to show that for any Kripke model $\mathbb{M}=(W,R,V)$, any assignment $g$ on $\mathbb{M}$, 

$$\mathbb{M},g\Vdash@_{\nomj}\beta\leq\alpha\mbox{ iff }\mathbb{M},g\Vdash\beta\leq\nomj\to\mathsf{A}\alpha.$$

Indeed, 
\begin{center}
\begin{tabular}{c l}
& $\mathbb{M},g\Vdash @_{\nomj}\beta\leq\alpha$\\
iff & $\mathbb{M},g\Vdash \neg\alpha\leq\neg @_{\nomj}\beta$\\
iff & $\mathbb{M},g\Vdash \neg\alpha\leq @_{\nomj}\neg\beta$\\
iff & $\mathbb{M},g\Vdash\mathsf{E}\neg\alpha\land\nomj\leq\neg\beta$\\
iff & $\mathbb{M},g\Vdash\beta\leq\nomj\to\mathsf{A}\alpha$.
\end{tabular}
\end{center}
\end{proof}

\begin{lemma}\label{Lemma:downarrow:substage2}
The residuation rules for $\downarrow$ are sound in $\mathbb{F}$.
\end{lemma}

\begin{proof}

For the left residuation rule for $\downarrow$, it suffices to show that for any Kripke model $\mathbb{M}=(W,R,V)$, any assignment $g$ on $\mathbb{M}$, 

$$\mathbb{M},g\Vdash\alpha\leq\downarrow x.\beta\mbox{ iff }\mathbb{M},g\Vdash\forall y(\mathsf{A}(y\to\alpha)\land y\leq\beta[y/x]).$$

Indeed,

\begin{center}
\begin{tabular}{c l}
& $\mathbb{M},g\Vdash \alpha\leq\downarrow x.\beta$\\
iff & for all $w\in W$, if $\mathbb{M},g,w\Vdash \alpha$, then $\mathbb{M},g,w\Vdash \downarrow x.\beta$\\
iff & for all $w\in W$, if $\mathbb{M},g^{y}_{w},w\Vdash \alpha$, then $\mathbb{M},g^{y}_{w},w\Vdash \downarrow x.\beta$\\
iff & for all $w\in W$, if $\mathbb{M},g^{y}_{w}\Vdash y\leq\alpha$, then $\mathbb{M},g^{y,x}_{w,w},w\Vdash\beta$\\
iff & for all $w\in W$, if $\mathbb{M},g^{y}_{w}\Vdash y\leq\alpha$, then $\mathbb{M},g^{y,x}_{w,w},w\Vdash\beta[y/x]$\\
iff & for all $w\in W$, if $\mathbb{M},g^{y}_{w}\Vdash \mathsf{A}(y\to\alpha)$, then $\mathbb{M},g^{y}_{w},w\Vdash\beta[y/x]$\\
iff & for all $w\in W$, if $\mathbb{M},g^{y}_{w}\Vdash \mathsf{A}(y\to\alpha)$, then $\mathbb{M},g^{y}_{w}\Vdash y\leq\beta[y/x]$\\
iff & for all $w\in W$, $\mathbb{M},g^{y}_{w}\Vdash \mathsf{A}(y\to\alpha)\land y\leq\beta[y/x]$\\
iff & $\mathbb{M},g\Vdash \forall y(\mathsf{A}(y\to\alpha)\land y\leq\beta[y/x])$,\\
\end{tabular}
\end{center}

where $y$ is a state variable that does not occur in $\alpha$ or $\beta$.




For the right residuation rule for $\downarrow$, it suffices to show that for any Kripke model $\mathbb{M}=(W,R,V)$, any assignment $g$ on $\mathbb{M}$, 

$$\mathbb{M},g\Vdash\downarrow x.\beta\leq\alpha\mbox{ iff }\mathbb{M},g\Vdash\forall y(\beta[y/x]\leq y\to\mathsf{E}(y\land\alpha)).$$

Indeed, 

\begin{center}
\begin{tabular}{c l}
& $\mathbb{M},g\Vdash\downarrow x.\beta\leq\alpha$\\
iff & $\mathbb{M},g\Vdash\neg\alpha\leq\neg\downarrow x.\beta$\\
iff & $\mathbb{M},g\Vdash\neg\alpha\leq\downarrow x.\neg\beta$\\
iff & $\mathbb{M},g\Vdash\forall y(\mathsf{A}(y\to\neg\alpha)\land y\leq\neg\beta[y/x])$\\
iff & $\mathbb{M},g\Vdash\forall y(\beta[y/x]\leq\neg(\mathsf{A}(y\to\neg\alpha)\land y))$\\
iff & $\mathbb{M},g\Vdash\forall y(\beta[y/x]\leq y\to\mathsf{E}(y\land\alpha))$.\\
\end{tabular}
\end{center}
\end{proof}

\begin{lemma}\label{Lemma:2ndsplitting:substage2}
The second splitting rule in Substage 2 is sound in $\mathbb{F}$.
\end{lemma}

\begin{proof}
It follows immediately from the meta-equivalence that $\forall x(\alpha\land\beta)\leftrightarrow\forall x\alpha\land\forall x\beta$.
\end{proof}

\begin{proposition}\label{Prop:Substage:3}
The packing rules in Substage 3 are sound in $\mathbb{F}$.
\end{proposition}

\begin{proof}
We only prove the soundness of the first packing rule, the other is similar.

For the left packing rule, to show its soundness, it suffices to show that for any Kripke model $\mathbb{M}=(W,R,V)$, any assignment $g$ on $\mathbb{M}$, 

$$\mathbb{M},g\Vdash\forall x(\alpha\leq\beta)\mbox{ iff }\mathbb{M},g\Vdash(\exists x\alpha)\leq\beta,$$

where $\beta$ does not contain occurrences of $x$. Indeed, 

\begin{center}
\begin{tabular}{c l}
& $\mathbb{M},g\Vdash\forall x(\alpha\leq\beta)$\\
iff & for all $w\in W$, $\mathbb{M},g^{x}_{w}\Vdash\alpha\leq\beta$\\
iff & for all $w,v\in W$, if $\mathbb{M},g^{x}_{w},v\Vdash\alpha$, then $\mathbb{M},g^{x}_{w},v\Vdash\beta$\\
iff & for all $w,v\in W$, if $\mathbb{M},g^{x}_{w},v\Vdash\alpha$, then $\mathbb{M},g,v\Vdash\beta$ (since $\beta$ does not contain occurrences of $x$)\\
iff & for all $v\in W$, if there exists a $w\in W$ such that $\mathbb{M},g^{x}_{w},v\Vdash\alpha$, then $\mathbb{M},g,v\Vdash\beta$\\
iff & for all $v\in W$, if $\mathbb{M},g,v\Vdash\exists x\alpha$, then $\mathbb{M},g,v\Vdash\beta$\\
iff & $\mathbb{M},g\Vdash(\exists x\alpha)\leq\beta$.\\
\end{tabular}
\end{center}
\end{proof}

\begin{proposition}\label{Prop:Substage:4}
The Ackermann rules in Substage 4 are sound in $\mathbb{F}$.
\end{proposition}

\begin{proof}
We only prove it for the right-handed Ackermann rule, the left-handed Ackermann rule is similar. Without loss of generality we assume that $n=m=1$. By the discussion on page \pageref{condition:1:4:equivalence}, it suffices to show the following \emph{right-handed Ackermann lemma}:
\begin{lemma}\label{Lemma:Right:Ackermann}
Assume $\alpha$ is pure and does not contain state variables in $\vec{x}$, $\beta$ is positive in $p$ and $\gamma$ is negative in $p$, then for any Kripke frame $\mathbb{F}=(W,R)$, any valuation $V$ and any assignment $g$ on it, the following are equivalent:

\begin{itemize}
\item $(\mathbb{F},V),g\Vdash\forall\vec{x}(\beta[\alpha/p]\leq\gamma[\alpha/p])$;
\item there exists a valuation $V^p$ such that $(\mathbb{F},V^{p}),g\Vdash\alpha\leq p$ and $(\mathbb{F},V^{p}),g\Vdash\forall\vec{x}(\beta\leq\gamma)$, where $V^{p}$ is the same as $V$ except $V^{p}(p)$.
\end{itemize}
\end{lemma}

$\Rightarrow$: Take $V^{p}$ such that $V^{p}$ is the same as $V$ except that $V^{p}(p)=\llbracket\alpha\rrbracket^{(\mathbb{F},V),g}$. Since $\alpha$ is pure, $\alpha$ does not contain $p$, it is easy to see that $\llbracket\alpha\rrbracket^{(\mathbb{F},V^{p}),g}=\llbracket\alpha\rrbracket^{(\mathbb{F},V),g}= \llbracket p\rrbracket^{(\mathbb{F},V^{p}),g}$. Therefore $(\mathbb{F},V^{p}),g\Vdash\alpha\leq p$. Then for any $w\in W$, any assignment $g'$ on $\mathbb{F}$ such that $g$ and $g'$ disagree at most at state variables in $\vec x$,

\begin{center}
\begin{tabular}{r l}
& $(\mathbb{F},V^{p}),g',w\Vdash p$\\
iff & $(\mathbb{F},V^{p}),g,w\Vdash p$\\
iff & $(\mathbb{F},V),g,w\Vdash\alpha$\\
iff & $(\mathbb{F},V),g',w\Vdash\alpha$,\\
\end{tabular}
\end{center}

so 
$$(\mathbb{F},V^{p}),g',w\Vdash\beta\mbox{ iff }(\mathbb{F},V),g',w\Vdash\beta[\alpha/p],$$ and 
$$(\mathbb{F},V^{p}),g',w\Vdash\gamma\mbox{ iff }(\mathbb{F},V),g',w\Vdash\gamma[\alpha/p],$$ so from $(\mathbb{F},V),g\Vdash\forall\vec{x}(\beta[\alpha/p]\leq\gamma[\alpha/p])$ one can get $(\mathbb{F},V^{p}),g\Vdash\forall\vec{x}(\beta\leq\gamma)$.\\

$\Leftarrow$: Assume that there exists a valuation $V^p$ such that $(\mathbb{F},V^{p}),g\Vdash\alpha\leq p$ and $(\mathbb{F},V^{p}),g\Vdash\forall\vec{x}(\beta\leq\gamma)$, where $V^{p}$ is the same as $V$ except $V^{p}(p)$. Then for any assignment $g'$ on $\mathbb{F}$ such that $g$ and $g'$ disagree at most at state variables in $\vec x$, 

since $\alpha$ and $p$ do not contain state variables in $\vec x$, from $$(\mathbb{F},V^{p}),g\Vdash\alpha\leq p$$ we have $$(\mathbb{F},V^{p}),g'\Vdash\alpha\leq p;$$

from $$(\mathbb{F},V^{p}),g\Vdash\forall\vec{x}(\beta\leq\gamma)$$ we have 

$$(\mathbb{F},V^{p}),g'\Vdash\beta\leq\gamma;$$ so from the monotonicity of $\beta$ in $p$ and the antitonicity of $\gamma$ in $p$ we get

$$(\mathbb{F},V^{p}),g'\Vdash\beta[\alpha/p]\leq\gamma[\alpha/p];$$

since $\beta[\alpha/p]$ and $\gamma[\alpha/p]$ do not contain $p$, we have 

$$(\mathbb{F},V),g'\Vdash\beta[\alpha/p]\leq\gamma[\alpha/p];$$

by the arbitrariness of $g'$ we have 

$$(\mathbb{F},V),g\Vdash\forall\vec{x}(\beta[\alpha/p]\leq\gamma[\alpha/p]).$$
\end{proof}

\section{Success of $\mathsf{ALBA}^{\downarrow}$}\label{Sec:Success}

In the present section we show that $\mathsf{ALBA}^{\downarrow}$ succeeds on all Sahlqvist inequalities.

\begin{theorem}\label{Thm:Success}
$\mathsf{ALBA}^{\downarrow}$ succeeds on all Sahlqvist inequalities.
\end{theorem}

\begin{definition}[Definite $\epsilon$-Sahlqvist inequality]
Given an order type $\epsilon$, $*\in\{-,+\}$, the signed generation tree $*\phi$ of the term $\phi(p_1,\ldots, p_n)$ is \emph{definite $\epsilon$-Sahlqvist} if there is no $+\lor,-\land$ occurring in the outer part on an $\epsilon$-critical branch. An inequality $\phi\leq\psi$ is definite $\epsilon$-Sahlqvist if the trees $+\phi$ and $-\psi$ are both definite $\epsilon$-Sahlqvist.
\end{definition}

\begin{lemma}\label{Lemma:Stage:1}
Let $\{\phi_i\leq\psi_i\}_{i\in I}=\mathsf{Preprocess}(\phi\leq\psi)$ obtained by exhaustive application of the rules in Stage 1 on an input $\epsilon$-Sahlqvist inequality $\phi\leq\psi$. Then each $\phi_i\leq\psi_i$ is a definite $\epsilon$-Sahlqvist inequality.
\end{lemma}

\begin{proof}
It is easy to see that by applying the distribution rules, all occurrences of $+\lor$ and $-\land$ in the outer part of an $\epsilon$-critical branch have been pushed up towards the root of the signed generation trees $+\phi$ and $-\psi$. Then by exhaustively applying the splitting rules, all such $+\lor$ and $-\land$ are eliminated. Since by applying the distribution rules, the splitting rules and the monotone/antitone variable elimination rules do not change the $\epsilon$-Sahlqvistness of a signed generation tree, in $\mathsf{Preprocess}(\phi\leq\psi)$, each signed generation tree $+\phi_i$ and $-\psi_i$ are $\epsilon$-Sahlqvist, and since they do not have $+\lor$ and $-\land$ in the outer part in the $\epsilon$-critical branches, they are definite.
\end{proof}

\begin{definition}[Inner $\epsilon$-Sahlqvist signed generation tree]
Given an order type $\epsilon$, $*\in\{-,+\}$, the signed generation tree $*\phi$ of the term $\phi(p_1,\ldots, p_n)$ is \emph{inner $\epsilon$-Sahlqvist} if its outer part $P_2$ on an $\epsilon$-critical branch is always empty, i.e.\ its $\epsilon$-critical branches have inner nodes only.
\end{definition}

\begin{lemma}\label{Lemma:Substage:2:1}
Given inequalities $\nomi_0\leq\phi_i$ and $\psi_i\leq\neg\nomi_1$obtained from Stage 1 where $+\phi_i$ and $-\psi_i$ are definite $\epsilon$-Sahlqvist, by applying the rules in Substage 1 of Stage 2 exhaustively, the inequalities that we get are in one of the following forms:

\begin{enumerate}
    \item pure inequalities which does not have occurrences of propositional variables;
    \item inequalities of the form $\nomi\leq\alpha$ or $x\leq\alpha$ where $+\alpha$ is inner $\epsilon$-Sahlqvist;
    \item inequalities of the form $\beta\leq\neg\nomi$ or $\beta\leq\neg x$ where $-\beta$ is inner $\epsilon$-Sahlqvist.
\end{enumerate}
\end{lemma}

\begin{proof}
Indeed, the rules in the Substage 1 of Stage 2 deal with outer nodes in the signed generation trees $+\phi_i$ and $-\psi_i$ except $+\lor$,$-\land$. For each rule, without loss of generality assume we start with an inequality of the form $\nomi\leq\alpha$, then by applying the rules in Substage 1 of Stage 2, the inequalities we get are either a pure inequality without propositional variables, or 
an inequality where the left-hand side (resp.\ right-hand side) is $\nomi$ or $x$ (resp.\ $\neg\nomi$ or $\neg x$), and the other side is a formula $\alpha'$ which is a subformula of $\alpha$, such that $\alpha'$ has one root connective less than $\alpha$. Indeed, if $\alpha'$ is on the left-hand side (resp.\ right-hand side) then $-\alpha'$ ($+\alpha'$) is definite $\epsilon$-Sahlqvist.

By applying the rules in the Substage 1 of Stage 2 exhaustively, we can eliminate all the outer connectives in the critical branches, so for non-pure inequalities, they become of form 2 or form 3.
\end{proof}

\begin{lemma}\label{Lemma:Substage:2:2}
Assume we have an inequality $\nomi\leq\alpha$ or $\beta\leq\neg\nomi$ where $+\alpha$ and $-\beta$ are inner $\epsilon$-Sahlqvist, by applying the rules in Substage 2 of Stage 2, we have (universally quantified) inequalities ($k$ can be 0 where a universally quantified inequality becomes an inequality) of the following form:

\begin{enumerate}
\item $\forall x_1\ldots\forall x_k(\alpha\leq p)$,
where $\epsilon(p)=1$, $\alpha$ is pure;

\item $\forall x_1\ldots\forall x_k(p\leq\beta)$,
where $\epsilon(p)=\partial$, $\beta$ is pure;

\item $\forall x_1\ldots\forall x_k(\alpha\leq\gamma)$,
where $\alpha$ is pure and $+\gamma$ is $\epsilon^{\partial}$-uniform;

\item $\forall x_1\ldots\forall x_k(\gamma\leq\beta)$,
where $\beta$ is pure and $-\gamma$ is $\epsilon^{\partial}$-uniform.
\end{enumerate}
\end{lemma}

\begin{proof}
First of all, from the rules of the Substage 2 of Stage 2, it is easy to see that from the given inequality, what we will obtain would be a set of mega-inequalities, and by applying the second splitting rule we would get universally quantified inequalities of the form $\forall x_1\ldots\forall x_k(\gamma\leq\delta)$. Now it suffices to check the shape of $\gamma$ and $\delta$. (From now on we call $\gamma\leq\delta$ the \emph{head} of the universally quantified inequality.)

Notice that for each input inequality, it is of the form $\nomi\leq\alpha$, $x\leq\alpha$ or $\beta\leq\neg\nomi$, $\beta\leq\neg x$, where $+\alpha$ and $-\beta$ are inner $\epsilon$-Sahlqvist. By applying the splitting rules and the residuation rules in this substage, it is easy to check that the head of the (universally quantified) inequality will have one side of the inequality pure, and the other side still inner $\epsilon$-Sahlqvist. By applying these rules exhaustively, one will either have $p$ as the non-pure side (with this $p$ on a critical branch), or have an inner $\epsilon$-Sahlqvist signed generation tree with no critical branch, i.e.,\ $\epsilon^{\partial}$-uniform.
\end{proof}

\begin{lemma}\label{Lemma:Substage:2:3}
Assume we have (universally quantified) inequalities of the form as described in Lemma \ref{Lemma:Substage:2:2}. Then we can get (universally quantified) inequalities of the following form:

\begin{enumerate}
\item $\alpha\leq p$ where $\epsilon(p)=1$, $\alpha$ is pure;
\item $p\leq\alpha$ where $\epsilon(p)=\partial$, $\alpha$ is pure;
\item $\forall x_1\ldots\forall x_k(\alpha\leq\gamma)$,
where $\alpha$ is pure and $+\gamma$ is $\epsilon^{\partial}$-uniform;

\item $\forall x_1\ldots\forall x_k(\gamma\leq\beta)$,
where $\beta$ is pure and $-\gamma$ is $\epsilon^{\partial}$-uniform.
\end{enumerate} 
\end{lemma}

\begin{proof}
For universally quantified inequalities of form 1 and 2 in Lemma \ref{Lemma:Substage:2:2}, we can apply the packing rule since $p$ does not contain occurrences of state variables. For universally quantified inequalities of form 3 and 4 in Lemma \ref{Lemma:Substage:2:2}, we do not need to apply any rules in this stage.
\end{proof}

\begin{lemma}\label{Lemma:Substage:2:4}
Assume we have (universally quantified) inequalities of the form as described in Lemma \ref{Lemma:Substage:2:3}, the Ackermann lemmas are applicable and therefore all propositional variables can be eliminated.
\end{lemma}

\begin{proof}
Immediate observation from the requirements of the Ackermann lemmas.
\end{proof}

\begin{proof}[Proof of Theorem \ref{Thm:Success}]
Assume we have an $\epsilon$-Sahlqvist inequality $\phi\leq\psi$ as input. By Lemma \ref{Lemma:Stage:1}, we get a set of definite  $\epsilon$-Sahlqvist inequalities. Then by Lemma \ref{Lemma:Substage:2:1}, we get inequalities as described in Lemma \ref{Lemma:Substage:2:1}. By Lemma \ref{Lemma:Substage:2:2}, we get the universally quantified inequalities as described. Therefore by Lemma \ref{Lemma:Substage:2:3}, we can apply the packing rules to get inequalities and universally quantified inequalities as described in the lemma. Finally by Lemma \ref{Lemma:Substage:2:4}, the (universally quantified) inequalities are in the right shape to apply the Ackermann rules, and thus we can eliminate all the propositional variables and the algorithm succeeds on the input.
\end{proof}

\section{Conclusion}\label{Sec:Conclusion}

In the present paper, we investigates the correspondence theory for hybrid logic with binder $\mathcal{H}(@,\downarrow)$. We define the class of Sahlqvist $\mathcal{H}(@,\downarrow)$-inequalities, and show that each of these inequalities has a first-order frame correspondent by an algorithm $\mathsf{ALBA}^{\downarrow}$. 

For future directions, we consider the canonicity theory for $\mathcal{H}(@,\downarrow)$, i.e.\ which class of $\mathcal{H}(@,\downarrow)$-formulas is preserved under taking canonical extensions or MacNeille completions, as well as develop the correspondence theory and canonicity theory for other very expressive hybrid languages, e.g.\ with $\exists, \Downarrow, \Sigma$ binders, where the $\exists$ binder is interpreted as in page \pageref{page:downarrow}, and $\Downarrow$ and $\Sigma$ binders are interpreted as follows:

\begin{center}
$\mathbb{M},g,w\Vdash \Sigma x.\varphi$ iff $\mathbb{M},g^{x}_{w'},w'\Vdash\varphi$ for some $w'\in W$;

$\mathbb{M},g,w\Vdash \Downarrow x.\varphi$ iff $\mathbb{M},g^{x}_{w},w'\Vdash\varphi$ for some $w'\in W$.
\end{center}

\paragraph{Acknowledgement} The research of the author is supported by Taishan University Starting Grant ``Studies on Algebraic Sahlqvist Theory'' and the Taishan Young Scholars Program of the Government of Shandong Province, China (No.tsqn201909151).

\bibliographystyle{abbrv}
\bibliography{Binder}

\end{document}